\documentclass[journal]{IEEEtran}
\usepackage{cite,amsmath,graphicx,amssymb,cite,amsthm,multirow}
\usepackage[ruled,lined,boxed,commentsnumbered]{algorithm2e}
\usepackage{algorithmic}
\usepackage{float}
\usepackage[table]{xcolor}
\usepackage{scalerel,stackengine}
\stackMath
\newcommand\reallywidehat[1]{%
\savestack{\tmpbox}{\stretchto{%
  \scaleto{%
    \scalerel*[\widthof{\ensuremath{#1}}]{\kern-.6pt\bigwedge\kern-.6pt}%
    {\rule[-\textheight/2]{1ex}{\textheight}}
  }{\textheight}%
}{0.5ex}}%
\stackon[1pt]{#1}{\tmpbox}%
}

\usepackage{color}

\newcounter{ALC@tempcntr}
\newcommand{\LCOMMENT}[1]{
    \setcounter{ALC@tempcntr}{\arabic{ALC@rem}}
    \setcounter{ALC@rem}{1}
    \item #1
    \setcounter{ALC@rem}{\arabic{ALC@tempcntr}}
}

\newtheorem{theorem}{Theorem}[section]
\newtheorem{corollary}{Corollary}[theorem]
\newtheorem{lemma}[theorem]{Lemma}
\newtheorem{assumption}{Assumption}%
\newtheorem{proposition}{Proposition}%

\begin{document}

\title{Invertible Particle Flow-based Sequential MCMC With Extension To Gaussian Mixture Noise Models}

\author{Yunpeng~Li,~Soumyasundar~Pal,
and~Mark~Coates
\thanks{Y. Li is with the Department of Computer Science, University of Surrey, Guildford, U.K., email: yunpeng.li@surrey.ac.uk;
S. Pal and M. Coates are with the Department of Electrical and Computer
Engineering, McGill University, Montr\'eal, Qu\'ebec, Canada,
e-mail: soumyasundar.pal@mail.mcgill.ca; mark.coates@mcgill.ca.
Y. Li and S. Pal are considered equal contributors. We acknowledge the support of the Natural Sciences and Engineering Research Council of Canada (NSERC) [2017-260250].}
}

\maketitle

\begin{abstract}

Sequential state estimation in non-linear and non-Gaussian state spaces has a wide range of applications
in statistics and signal processing. One of the most
effective non-linear filtering approaches, particle filtering, suffers from
weight degeneracy in high-dimensional filtering scenarios.
Several avenues have been pursued to address
high-dimensionality. Among these,
particle flow particle filters construct effective proposal distributions by using
invertible flow to migrate particles continuously
from the prior distribution to the posterior, and
sequential Markov chain Monte Carlo (SMCMC) methods
use a Metropolis-Hastings (MH) accept-reject approach to improve filtering performance.
In this paper, we propose to combine the strengths
of invertible particle flow and SMCMC
by constructing a composite Metropolis-Hastings (MH) kernel
within the SMCMC framework using invertible particle
flow. In addition, we propose a Gaussian mixture model (GMM)-based particle flow algorithm to construct effective MH kernels for multi-modal distributions.
Simulation results show that  for high-dimensional state estimation
example problems the proposed kernels
significantly increase the acceptance rate with minimal additional computational overhead
and improve estimation accuracy compared with
state-of-the-art filtering algorithms.
\end{abstract}

\IEEEpeerreviewmaketitle

\section{Introduction}
\label{sec:intro}

Effective and efficient on-line learning of high-dimensional states is
an important task in many domains where we need to regularly update
our knowledge by processing a deluge of streaming data.  Relevant applications span
from robotic learning to financial modelling, from multi-target
tracking to weather
forecasting~\cite{martinez2009,creal2015,nannuru2013,vanLeeuwen2015}.
In the presence of non-linear states space models, particle
filters~\cite{gordon1993,doucet2009} are one of the standard tools
for sequential inference. However, it is in general difficult to
construct effective proposal distributions to match well with the
posterior distribution in high-dimensional spaces.  A mismatch between
the proposal and the posterior leads to negligible
importance weights for most particles. These are
normally discarded after resampling.  This weight degeneracy issue
thus leads to poor particle representation of the posterior and has limited
the widespread application of particle filters in high-dimensional filtering
scenarios~\cite{bengtsson2008,snyder2008,beskos2014}.

Advanced particle filtering methods have been proposed
to approximate the optimal proposal distribution~\cite{doucet2000b} in order to alleviate the weight degeneracy issue.
The~\textsl{multiple particle filter}~\cite{djuric2007}
partitions the state space into multiple lower-dimensional spaces where state estimation is performed.
The~\textsl{block particle filter} takes a similar approach by partitioning
the state space and the measurement space into independent blocks in updating the filtering distribution,
although this introduces a bias that is difficult to quantify~\cite{rebeschini2015}.
In~\cite{beskos2017}, Beskos et al.\ introduced an unbiased~\textsl{space-time particle filter} which also relies on factorization of the conditional posterior. These algorithms
are promising but are applicable only
in scenarios where the factorization can be performed.
The~\textsl{equivalent weights particle
filter} sacrifices the statistical consistency
to ensure substantial weights for a large number of particles~\cite{ades2015}.
{\em Particle flow filters} have been proposed to address high-dimensional filtering
by transporting particles continuously from the prior to the posterior~\cite{daum2007,ding2012,daum2014b,daum2017}. Most particle flow filters do not provide statistically consistent estimates of the posterior. Deterministic particle flow filters such as those we build on in this paper often underestimate the variance. Other versions incorporate approximation errors in the implementation
or impose overly strong model assumptions~\cite{daum2010a,khan2015}. Recent stochastic particle flow filters such as those described in~\cite{daum2017} address these limitations. As an alternative approach, \cite{reich2013,heng2015,bunch2016,li2017} propose the use of particle flow or transport maps to construct a proposal distribution that approximates the posterior
and then use importance sampling to correct for any discrepancy.

Another direction to combat weight degeneracy in high-dimensional filtering is to use
Markov chain Monte Carlo (MCMC) methods to improve the diversity of samples.
MCMC is often considered as the most
effective method for sampling from a high-dimensional distribution~\cite{andrieu2003}.
More effective MCMC techniques in high-dimensional spaces use Hamiltonian or Langevin 
dynamics to construct efficient proposals~\cite{duane1987,neal2011,girolami2011,welling2011}.
Although effective, these techniques cannot be directly used
in sequential inference tasks, as MCMC typically targets a static distribution.
To use MCMC to diversify samples in a sequential setting, the resample-move particle filter incorporates
MCMC methods by performing MCMC moves after the resampling step of the particle filter~\cite{gilks2001}.
Unfortunately, the resampling step can lead to degeneracy, so many
MCMC moves may be required to diversify particles.

A more general framework called sequential Markov chain Monte Carlo
(SMCMC) uses MCMC techniques to sample directly from the approximate target
distribution~\cite{khan2005,septier2009,septier2016,li2017b}.
\cite{khan2005} directly targets the filtering distribution which is
computationally expensive.  \cite{septier2009,septier2016, li2017b} instead propose to sample from the joint state
distribution to avoid the numerical integration of the predictive
density.  In the algorithms presented in~\cite{septier2009,
  septier2016, li2017b}, a fixed number of samples is used to
approximate the empirical posterior distribution for each time
step. By contrast, in the sequentially interacting MCMC (SIMCMC) framework described in~\cite{brockwell2010},
one can continue to generate interacting non-Markovian samples of the entire state sequence to improve the empirical approximation of joint posterior distribution successively. The resulting samples are asymptotically distributed according to the joint posterior distribution. The fundamental difference between the SMCMC and the SIMCMC techniques is that the SMCMC algorithm consists of sequential implementation of static MCMC schemes (justifying the name ``sequential MCMC"), whereas this interpretation does not hold for the SIMCMC algorithm. Hence the analysis of SIMCMC in~\cite{brockwell2010} cannot be applied to SMCMC (and vice-versa) and SMCMC cannot be expressed as a special case of SIMCMC. From a practical viewpoint, if a fixed number of particles is to be used, the effect of error in approximating the posterior distribution at the previous time step might be severe for the SIMCMC algorithm and limit its applicability in high dimensional online filtering problems compared to advanced SMC or SMCMC techniques.      


A variety of MCMC kernels developed for
sampling in high-dimensional spaces can be used inside the SMCMC
framework.~\cite{li2017b} and this paper differ from the past
literature by using particle flow to construct MCMC kernels in the
SMCMC framework.


In this paper, we propose to incorporate invertible particle
flow~\cite{li2017} into the SMCMC framework. Our main contributions
are three-fold: 1) we exploit the capability of particle flow 
to migrate samples into high posterior density regions to construct a
composite Metropolis-Hasting (MH) kernel that significantly increases
the acceptance rate of the joint draw; 2) for multi-modal
distributions, we incorporate a Gaussian mixture model-based particle
flow to improve sampling efficiency; 3) we assess the performance of
the proposed methods through numerical simulation of challenging high-dimensional filtering
examples.  We presented preliminary results concerning an initial attempt to incorporate SMCMC with invertible particle flow
in~\cite{li2017b}. This paper provides a more detailed description
of the proposed approach, presents more computationally efficient
algorithms, and proposes a sequential MCMC method with mixture
model-based flow for efficient exploration of multi-modal
distributions. We also provide theoretical results regarding
asymptotic convergence of the algorithms. These results are not restricted to the
invertible particle flow composite kernel case, but hold for SMCMC methods provided the kernel and
filtering problem satisfy certain assumptions (see Section~\ref{sec:convergence}). While similar in
spirit to results presented in~\cite{brockwell2010} for the SIMCMC algorithm, the key assumptions
are slightly less restrictive and the theorem directly addresses the
sequential implementation in~\cite{septier2016} and this work. Recently, \cite{finke2018} carries out a rigorous statistical analysis of SMCMC algorithms to establish upper bounds on finite sample filter errors, in addition to providing asymptotic convergence results.

The rest of the paper is organized as follows. Section~\ref{sec:ps} provides the problem statement and
Section~\ref{sec:related_work} reviews the sequential MCMC framework and the invertible particle flow. We present 
the proposed methods in Section~\ref{sec:method}.
Simulation examples and results are presented in Section~\ref{sec:result}.
We conclude the paper in Section~\ref{sec:conclusion}.

\section{Problem Statement}
\label{sec:ps}

Many online data analysis applications involve estimating unknown
quantities, which we call the ``state'' $x$, given sequential measurements.
Often there is prior knowledge related to the state $x_0$, where
the subscript $0$ indicates the time step before any measurement arrives.
A dynamic model describes the evolution of state $x_k \in \mathbb{R}^{d}$
at time step $k$ given the past states and a measurement model captures the relationship between observations $z_k \in \mathbb{R}^{S}$ and
the state $x_k$.

The hidden Markov model (HMM) is a ubiquitous tool for modelling
the discrete-time dynamic system and measurement process.
It is assumed that the hidden state $x_k$ follows the Markov property,
i.e., it is independent of all states before $k{-}1$ given the
state $x_{k-1}$.
The measurement $z_k$ is modelled as independent of all past measurements and past states conditioned on the current state $x_k$.
We can model the state evolution and measurements
with the following HMM:
\begin{align}
x_0 &\sim p(x_0)\label{eq:prior}\,\,\,,\\
x_k &= g_k(x_{k-1},v_k)\quad \mbox{for~} k\geq 1 \label{eq:dynamic}\,\,,\\
z_k &= h_k(x_k,w_k)\quad \mbox{for~} k\geq 1 \label{eq:measurement}\,.
\end{align}
Here $p(x_0)$ is the initial probability density function of the state $x_0$,
$g_k: \mathbb{R}^{d} \times \mathbb{R}^{d'} \rightarrow
\mathbb{R}^{d}$ models the dynamics of the unobserved state
$x_k$, and the measurement model $h_k:\mathbb{R}^{d} \times \mathbb{R}^{S'} \rightarrow \mathbb{R}^{S}$ describes the relation between the measurement $z_k$ and the state $x_k$.  We assume that
$h_k(x_k,0)$ is a $C^1$ function,
i.e., $h_k(x_k,0)$ is a differentiable function whose
first derivative is continuous.
$v_k \in \mathbb{R}^{d'}$ and $w_k\in \mathbb{R}^{S'}$ are the process noise and the measurement noise, respectively.
With these models,
the filtering task is to compute the posterior distribution of the state trajectory
$p(x_{k}|z_{1},\ldots,z_{k})$ online, as new data become available.
For conciseness, we use $x_{a:b}$ to denote the set $\{x_a,x_{a+1},\ldots,x_b\}$
and $z_{a:b}$ to denote the set $\{z_a,z_{a+1},\ldots,z_b\}$,
where $a$ and $b$ are integers and $a < b$.
The posterior distribution $p(x_{k}|z_{1:k})$ gives a probabilistic interpretation
of the state given all measurements, and can be used for state
estimation and detection.
%

\section{Background Material}
\label{sec:related_work}

\subsection{Sequential Markov chain Monte Carlo methods}
\label{sec:SMCMC}

Sequential Markov chain Monte Carlo (SMCMC) methods were
proposed as a general MCMC approach for approximating the joint posterior distribution $\pi_k(x_{0:k}) = p(x_{0:k}|z_{1:k})$ recursively.
A unifying framework of the various SMCMC methods was provided in~\cite{septier2016}.
At time step $k$, $\pi_k(x_{0:k})$ can be computed pointwise up to a constant in a
recursive manner:
\begin{align}
\pi_k(x_{0:k}) &= p(x_{0:k}|z_{1:k}) \propto p(x_{0:k},z_{1:k}) \,,\nonumber\\
&\propto p(x_k|x_{k-1})p(z_k|x_k)p(x_{0:k-1}|z_{1:k-1})p(z_{1:k-1})\,,\nonumber\\
&\propto  p(x_k|x_{k-1})p(z_k|x_k)\pi_{k-1}(x_{0:k-1})\,.
\label{eq:recursive_joint_two} 
\end{align}

As $\pi_{k-1}(x_{0:k-1})$ is not analytically tractable, it is
impossible to sample from it in a general HMM.
In all SMCMC methods, the distribution is replaced by its empirical
approximation in (\ref{eq:recursive_joint_two}), which leads to an approximation
of $\pi_k$ as follows:
\begin{align}
\breve \pi_k(x_{0:k}) \propto p(x_k|x_{k-1})p(z_k|x_k)\widehat \pi_{k-1}(x_{0:k-1})\,,
\label{eq:approx_posterior} 
\end{align}
where
\begin{align}
\widehat \pi_{k-1}(x_{0:k-1}) = \dfrac{1}{N_p}\sum_{j=N_b+1}^{N_b+N_p}\delta_{x_{k-1,0:k-1}^{j}}(x_{0:k-1})\,.
\label{eq:approx_k_1}
\end{align}
Here $\delta_{a}(\cdot)$ is the Dirac delta function centred at $a$,
$N_b$ is the number of samples discarded during a burn-in period, and $N_p$ is the number of
retained MCMC samples. $\{x_{k-1,0:k-1}^{j}\}_{j=N_b+1}^{N_b+N_p}$ are
the $N_p$ samples obtained from the Markov chain at time $k-1$, whose
stationary distribution is $\breve \pi_{k-1}(x_{0:k-1})$. At time step
$k$, $N_b+N_p$ iterations of the Metropolis-Hastings (MH) algorithm with proposal $q_k(\cdot)$
are executed to
generate samples $\{x_{k,0:k}^{j}\}_{j=N_b+1}^{N_b+N_p}$ from the invariant distribution $\breve \pi_k(x_{0:k})$, and $\pi_k(x_{0:k})$ is approximated as:
\begin{align}
\widehat \pi_k(x_{0:k}) = \dfrac{1}{N_p}\sum_{j=N_b+1}^{N_b+N_p}\delta_{x_{k,0:k}^{j}}(x_{0:k})
\label{eq:empirical_posterior}
\end{align}
The purpose of the joint draw of $\breve{\pi}_k(x_{0:k})$ is to avoid numerical integration
of the predictive density when the target distribution
is $p(x_{k}|z_{1:k})$~\cite{septier2009}.
Note that if we are only interested in approximating
the marginal posterior distribution $p(x_{k}|z_{1:k})$,
only $\{x_{k-1,k-1}^{j}\}_{j=N_b+1}^{N_b+N_p}$ needs to be stored instead of the full past state trajectories
$\{x_{k-1,0:k-1}^{j}\}_{j=N_b+1}^{N_b+N_p}$. The
Metropolis-Hastings (MH) algorithm used within SMCMC to generate one
sample is summarized in Algorithm \ref{alg:vanilla_SMCMC}.

\begin{algorithm}[h]
\begin{algorithmic}[1]
\STATE Propose $x_{k,0:k}^{*(i)} \sim q_{k}(x_{0:k}|x_{k,0:k}^{i-1})$\;
\STATE Compute the MH acceptance probability
$\rho = \min\big(1,\frac{\breve \pi_k(x_{k,0:k}^{*(i)})}
{q_{k}(x_{k,0:k}^{*(i)}|x_{k,0:k}^{i-1})}
\frac{q_{k}(x_{k,0:k}^{i-1}|x_{k,0:k}^{*(i)})}
{\breve \pi_k(x_{k,0:k}^{i-1})}
\big)$\;
\STATE Accept $x_{k,0:k}^{i} = x_{k,0:k}^{*(i)}$ with probability $\rho$, otherwise \\set $x_{k,0:k}^{i} = x_{k,0:k}^{i-1}$\;
\end{algorithmic}
\caption{MH Kernel in SMCMC~\cite{septier2016}.
\newline Input: $x_{k,0:k}^{i-1}$.
\newline Output: $x_{k,0:k}^{i}$.}
\label{alg:vanilla_SMCMC}
\end{algorithm}

\subsubsection{Composite MH kernel in SMCMC}

Different choices of the MCMC kernel for high dimensional SMCMC
are discussed in~\cite{septier2016}. In most SMCMC algorithms, an independent MH kernel is adopted~\cite{septier2016},
i.e., $q_k(x_{0:k-1}|x_{k,0:k}^{i-1}) = q_k(x_{0:k})$, meaning that
the proposal is independent of the state of the Markov chain at the
previous iteration.
The ideal choice is the optimal independent MH kernel,
but it is usually impossible to sample from~\cite{septier2016}.
It is difficult to identify an effective approximation to the optimal independent MH kernel. The choice of independent MH kernel using the prior as the proposal can lead to very low acceptance rates if the state dimension is very high or the measurements are highly informative.

\cite{septier2009,septier2016} propose the use of a composite MH kernel which is constituted of a joint
proposal $q_{k,1}$ (to update $x_{0:k}$) followed by two individual
state variable refinements using proposals $q_{k,2}$ (to update
$x_{0:k-1}$) and $q_{k,3}$ (to update $x_{k}$), based on the \emph{Metropolis within Gibbs} approach,
within a single MCMC iteration. The composite kernel is summarized in
Algorithm~\ref{alg:composite_kernel}.  

\begin{algorithm}[h]
\begin{algorithmic}[1]
\LCOMMENT {\underline{Joint draw of $x_{k,0:k}^{i}$}:}
\STATE Propose $x_{k,0:k}^{*(i)} \sim q_{k,1}(x_{0:k}|x_{k,0:k}^{i-1})$\;
\STATE Compute the MH acceptance probability
$\rho_1= \min\big(1,\frac{\breve \pi_k(x_{k,0:k}^{*(i)})}
{q_{k,1}(x_{k,0:k}^{*(i)}|x_{k,0:k}^{i-1})}
\frac{q_{k,1}(x_{k,0:k}^{i-1}|x_{k,0:k}^{*(i)})}
{\breve \pi_k(x_{k,0:k}^{i-1})}
\big)$\;
\STATE Accept $x_{k,0:k}^{i} = x_{k,0:k}^{*(i)}$ with probability $\rho_1$, otherwise \\set $x_{k,0:k}^{i} = x_{k,0:k}^{i-1}$\;
\LCOMMENT {\underline{Individual refinement of $x_{k,0:k-1}^{i}$}:}
\STATE Propose $x_{k,0:k-1}^{*(i)} \sim q_{k,2}(x_{0:k-1}|x_{k,0:k}^{i})$\;
\STATE Compute the MH acceptance probability
$\rho_2 = \min\big(1,\frac{\breve \pi_k(x_{k,0:k-1}^{*(i)},x_{k,k}^{i})}
{q_{k,2}(x_{k,0:k-1}^{*(i)}|x_{k,0:k}^{i})}
\frac{q_{k,2}(x_{k,0:k-1}^{i}|x_{k,0:k-1}^{*(i)},x_{k,k}^{i})}{\breve \pi_k(x_{k,0:k}^{i})}
\big)$\;
\STATE Accept $x_{k,0:k-1}^{i} = x_{k,0:k-1}^{*(i)}$ with probability $\rho_2$\;
\LCOMMENT {\underline{Individual refinement of $x_{k,k}^{i}$}:}
\STATE Propose $x_{k,k}^{*(i)} \sim q_{k,3}(x_k|x_{k,0:k}^{i})$\;
\STATE Compute the MH acceptance probability
$\rho_3 = \min\big(1,\frac{\breve \pi_k(x_{k,0:k-1}^{i},x_{k,k}^{*(i)})}
{q_{k,3}(x_{k,k}^{*(i)}|x_{k,0:k}^{i})}
\frac{q_{k,3}(x_{k,k}^{i}|x_{k,0:k-1}^{i},x_{k,k}^{*(i)})}{\breve \pi_k(x_{k,0:k}^{i})}
\big)$\;
\STATE Accept $x_{k,k}^{i} = x_{k,k}^{*(i)}$ with probability
$\rho_3$\;
\end{algorithmic}
\caption{Composite MH Kernels
in a unifying framework of SMCMC~\cite{septier2009,septier2016}.
\newline Input: $x_{k,0:k}^{(i-1)}$
\newline Output: $x_{k,0:k}^{(i)}$}
\label{alg:composite_kernel}
\end{algorithm}

Any of the MCMC kernels mentioned before can be used in the joint draw
step of a composite kernel. For example, the independent MH kernel
based on the prior as the proposal is used in the joint draw step of the implementation of the sequential manifold Hamiltonian Monte Carlo (SmHMC) algorithm in~\cite{septier2016}. For individual refinement of $x_{k,0:k-1}^i$, \cite{septier2016} uses the independent proposal $q_{k,2} = \widehat{\pi}_{k-1}$, which leads to the following simplification of MH acceptance rate in Line 5 of Algorithm \ref{alg:vanilla_SMCMC}, using Equation \eqref{eq:approx_posterior}.

\begin{align}
\rho_2 &= \min\big(1,\frac{\breve \pi_k(x_{k,0:k-1}^{*(i)},x_{k,k}^{i})\widehat{\pi}_{k-1}(x_{k,0:k-1}^{i})}
{\widehat{\pi}_{k-1}(x_{k,0:k-1}^{*(i)})\breve \pi_k(x_{k,0:k}^{i})}\big)\,\nonumber\\
& = \min\big(1,\frac{p(x_{k,k}^i|x_{k,k-1}^{*(i)})}{p(x_{k,k}^i|x_{k,k-1}^i)}\big)\,.
\label{eq:rho_2_calculation}
\end{align}

The aim of the refinement steps is to explore the neighbourhood of
samples generated in the joint draw step. For the MCMC kernel of the
individual refinement step of $x_k$, Langevin diffusion or Hamiltonian
dynamics have been proposed to more efficiently traverse a
high-dimensional space~\cite{septier2016}.  The manifold Hamiltonian
Monte Carlo (mHMC) kernel $q_{k,3}(\cdot)$ of the individual
refinement step of the SmHMC algorithm efficiently samples from the
target filtering distribution, making the SmHMC algorithm one of the
most effective algorithms for filtering in high-dimensional spaces.


\subsection{Particle flow particle filter}

\subsubsection{Particle flow filter}

In the last decade, a new class of Monte Carlo-based filters has emerged that shows promise in high-dimensional filtering.
In a given time step $k$, particle flow algorithms~\cite{daum2007,daum2010a}
migrate particles from the predictive distribution $p(x_k|z_{1:k-1})$
to the posterior distribution $p(x_k|z_{1:k})$, via a ``flow'' that
is specified through a partial differential equation (PDE).
There is no sampling (or resampling).
Thus the weight degeneracy issue is avoided.

A particle flow can be modelled by a background stochastic process
$\eta_\lambda$ in a pseudo-time interval $\lambda \in [0, 1]$,
such that the distribution of $\eta_0$ is the predictive distribution
$p(x_k|z_{1:k-1})$ and the distribution
of $\eta_1$ is the posterior distribution
$p(x_k|z_{1:k}) = \dfrac{p(x_k|z_{1:k-1})p(z_k|x_k)}{p(z_k|z_{1:k-1})}$.

In~\cite{daum2010a}, the underlying stochastic process $\eta_\lambda$ satisfies 
an ordinary differential equation (ODE) with zero diffusion:
\begin{align}
\dfrac{d\eta_\lambda}{d\lambda} = \varphi(\eta_\lambda,\lambda)\,\,.
\label{eq:dynamic_no_fusion}
\end{align}
When the predictive distribution and the
likelihood are both Gaussian, the exact flow for the linear Gaussian model is:
\begin{align}
\varphi(\eta_\lambda,\lambda)=\dfrac{d\eta_\lambda}{d\lambda}=A(\lambda)\eta_\lambda+b(\lambda)\,\,,
\label{eq:gauss_flow_solution_EDH}
\end{align} 
where
\begin{align}
A(\lambda)&=-\dfrac{1}{2}PH^T(\lambda HPH^T+R)^{-1}H\,\,,\label{eq:EDH_A_linear}\\
b(\lambda)&=(I+2\lambda A(\lambda))[(I+\lambda A(\lambda))PH^TR^{-1}z+A(\lambda)\bar{\eta}_0]\,\,,\label{eq:EDH_b_linear}
\end{align}
Here $\bar{\eta}_0$ is the mean of the predictive distribution, $P$ is the covariance matrix of
prediction error for
the prior distribution, $z$ is the new measurement,
$H$ is the measurement matrix, i.e. $h_k(x_k) = Hx_k$,
and $R$ is the covariance matrix of the measurement noise.
We refer to this method as the exact Daum-Huang (EDH) filter,
and a detailed description of its implementation is provided
in~\cite{choi2011}.
For nonlinear models, a computationally intensive variation
of EDH, that computes a separate flow 
for each particle by performing linearization at the particle location $\eta_{\lambda}^i$, was proposed in~\cite{ding2012} and is referred to as the localized exact Daum-Huang (LEDH) filter.


Numerical integration is normally used to solve the ODE
in Equation~\eqref{eq:dynamic_no_fusion}. The integral between $\lambda_{j-1}$ and $\lambda_{j}$ for $1 \leq j \leq N_\lambda$, where $\lambda_{0} = 0$
and $\lambda_{N_\lambda}=1$,
is approximated and the Euler update rule
for the EDH flow becomes 
\begin{align}
\eta_{\lambda_{j}}^{i} &= f_{\lambda_j}(\eta_{\lambda_{j-1}}^i)\nonumber\\
& = \eta_{\lambda_{j-1}}^{i}
+ \epsilon_j(A(\lambda_{j})\eta_{\lambda_{j-1}}^i
+ b(\lambda_{j}))\,,
\label{eq:discrete_update}
\end{align}
where the step size $\epsilon_j = \lambda_j - \lambda_{j-1}$ and $\displaystyle\sum_{j=1}^{N_\lambda} \epsilon_j =1$.

\subsubsection{Particle flow particle filter}
Because of the discretization errors made while numerically solving Eq \eqref{eq:dynamic_no_fusion}, and 
the mismatch of modelling assumptions between a general HMM and a linear Gaussian setup (which was assumed in deriving Equations \eqref{eq:EDH_A_linear} and \eqref{eq:EDH_b_linear}),
the migrated particles after the particle flow process are
not exactly distributed according to the posterior. Instead  $\eta_1^i$ can be viewed as being drawn from a proposal distribution
$q(\eta_1^i|x_{k-1}^i,z_k)$, which is possibly well matched to the
posterior, because of the flow procedure. It is shown in
\cite{li2017} that for the EDH, if an auxiliary flow is performed starting from the mean of the predictive distribution 
$\bar{\eta}_{0}$, and the generated flow parameters are used to perform particle flow for each particle $\eta_0^i$, then in presence of the assumed smoothness condition on the
measurement function $h$ and with sufficiently small step sizes
$\epsilon_j$, the auxiliary particle flow establishes an
invertible mapping $\eta_1^i =T(\eta_0^i;z_k,x_{k-1}^i)$. We can straightforwardly compute the proposal density as follows:
\begin{align}
q(\eta_1^i|x_{k-1}^i,z_k) &= \dfrac{p(\eta_0^i|x_{k-1}^i)}{|\det(\dot{T}(\eta_0^i;x_{k-1}^i,z_k))|},\,
\end{align} 
$\dot{T}(\cdot) \in \mathbb{R}^{d\times d}$
is the Jacobian determinant of the mapping function $T(\cdot)$ and $|\cdot|$ denotes the absolute value. The determinant of $\dot{T}(\cdot)$ is given as:
\begin{align}
\det(\dot{T}(\eta_0^i;x_{k-1}^i,z_k)) = \prod^{N_{\lambda}}_{j=1}\det(I+ \epsilon_jA(\lambda_j))\;
\end{align}


\subsubsection{Particle flow particle filter with Gaussian mixture model assumptions}
\label{sec:PFPF_GMM}

When the state space models involve Gaussian mixture
model (GMM) noise, Equations~\eqref{eq:dynamic} and \eqref{eq:measurement}
admit the following structures
$v_k \sim
\Sigma^M_{m=1}\alpha_{k,m}\mathcal{N}(\psi_{k,m},Q_{k,m})$ and $w_k \sim
\Sigma^N_{n=1}\beta_{k,n}\mathcal{N}(\zeta_{k,n},R_{k,n})$.

An alternative representation of the GMM allows
an equivalent formulation of the dynamic model
by introducing a latent scalar variable $d_k \in \{1,2,...M\}$ with a probability mass function (PMF)
$p(d_k\!=\!m)=\alpha_{k,m}$. The variable
$d_k$ specifies the GMM component that excites the dynamic model, i.e.,
$p(x_k|x_{k-1},d_k\!=\!m) =\mathcal{N}(x_k|g_k(x_{k-1}) +
\psi_{k,m},Q_{k,m})$.
The state transition density can then be described as:
\begin{align}
\label{switched_process}
& p(x_k|x_{k-1}) =\textstyle{\sum}^M_{m=1}p(d_k=m)p(x_k|x_{k-1},d_k=m)\nonumber\\
&\quad \quad \,\, =\textstyle{\sum}^M_{m=1}\alpha_{k,m}\mathcal{N}(x_k|g_k(x_{k-1}) + \psi_{k,m},Q_{k,m}) \,\,.
\end{align}
Similarly, a latent variable $c_k$ with the PMF $p(c_k\!=\!n)=\beta_{k,n}$ specifies the GMM component that
generates the measurement noise.
The likelihood can be described as follows:
\begin{align}
\label{switched_measuerement}
p(z_k|x_k) &= \textstyle{\sum}^N_{n=1}p(c_k=n)p(z_k|x_k,c_k=n)\nonumber\\
&=\textstyle{\sum}^N_{n=1}\beta_{k,n} \mathcal{N}(z_k|h_k(x_k)+\zeta_{k,n},R_{k,n})\,\,.
\end{align} 

In order to construct the particle flow particle filter for this model, at time step $k$, first $d_k^i = m \in\{1,2,...M\}$ is sampled with probability $\{\alpha_{k,1},\alpha_{k,2},...\alpha_{k,M}\}$ and
$c_k^i=n \in\{1,2,...N\}$ is sampled with probability $\{\beta_{k,1},\beta_{k,2},...\beta_{k,N}\}$.
Conditioned on $(d_k^i,c_k^i)$, an invertible particle flow $(A_{mn}^i(\lambda),b_{mn}^i(\lambda))$ based on LEDH is constructed using the
 $m$-th and $n$-th Gaussian components of the dynamic and measurement models respectively, to sample $x_k^i$. The importance weights of the
joint state $\{x_k,d_k,c_k\}$ can be updated as follows~\cite{pal2018}:
\begin{align}
\omega_k^i &\propto \omega_{k-1}^i \dfrac{p(x_k^i|x_{k-1}^i,d_k^i)p(z_k|x_k^i,c_k^i)}{q(x_k^i|x_{k-1}^i,d_k^i,c_k^i,z_k)}\,
\end{align}
where the proposal density is computed by:
\begin{align}
q(x_k^i|x_{k-1}^i,d_k^i\!=\!m,c_k^i\!=\!n,z_k) =
\dfrac{p(\eta_0^i|x_{k-1}^i,d_k^i=m)}{|\displaystyle\prod^{N_{\lambda}}_{j=1}\det(I+\epsilon_j A_{mn}^i(\lambda_j))|}\,\,.
\label{eq:proposal_GMM}
\end{align}

\section{Methods}
\label{sec:method}

In this section, we propose to use the invertible particle flow to approximate
the optimal independent MH kernel in the sequential MCMC methods for both uni-modal
and multi-modal target distributions.

\subsection{SMCMC with invertible particle flow}
\label{sec:SMCMC_flow}

To construct MH kernels based on invertible particle flow, we first develop a new formulation of the invertible
mapping with particle flow.

\subsubsection{New formulation of the invertible mapping with particle flow}

The particle flow particle filters (PF-PFs)
construct invertible particle flows in a pseudo-time interval
$\lambda \in [0, 1]$ in order to move particles
drawn from the prior distribution into regions
where the posterior density is high.

Using the Euler update rule specified in Equation~\eqref{eq:discrete_update} recursively over $j = N_{\lambda},N_{\lambda}-1, \cdots, 2,1$,
the invertible mapping for the PF-PF (EDH) can be expressed as:
\begin{align}
\eta_1^i =& f_{\lambda_{{N_\lambda}}}(f_{\lambda_{{N_\lambda-1}}}(\ldots f_{\lambda_1}(\eta_0^i))\nonumber\\
=& (I+\epsilon_{N_\lambda} A(\lambda_{N_\lambda}))\eta_{\lambda_{N_\lambda-1}}^i
+\epsilon_{N_\lambda}b(\lambda_{N_\lambda})\nonumber\\
=&\ldots\nonumber\\
=&C \eta_{0}^i + D\,\,,
\label{eq:EDH_mapping}
\end{align}
where
\begin{align}
C=\prod_{j=1}^{N_\lambda}(I+\epsilon_{N_\lambda+1-j} A(\lambda_{N_\lambda+1-j}))\,\,,
\label{eq:C}
\end{align}
and
\begin{align}
D &= \epsilon_{N_\lambda}b(\lambda_{N_\lambda})+\nonumber\\
&\sum_{m=1}^{N_\lambda-1}([\prod_{j=1}^{N_\lambda-m}
(I+\epsilon_{N_\lambda+1-j} A(\lambda_{N_\lambda+1-j}))]
\epsilon_{j}b(\lambda_j))\,\,.
\label{eq:D}
\end{align}

In~\cite{li2017} it is shown
that the equivalent mapping is invertible with sufficiently small
$\epsilon_j$, so
the matrix $C$ is invertible.
The procedure to produce $C$ and $D$ is summarized in
Algorithm~\ref{alg:flow_EDH_parameters} and the proposal density becomes
\begin{align}
q(\eta_1^i|x_{k-1}^i,z_k) =\dfrac{p(\eta_0^i|x_{k-1}^i)}{|\det(C)|}\,\,.
\label{eq:PFPF_EDH_proposal}
\end{align}


\begin{algorithm}[h]
\begin{algorithmic}[1]
\STATE Initialize: $C = I, D = \mathbf{0}$\;
\FOR {$j = 1,\ldots,N_\lambda$}
\STATE Set $\lambda_j = \lambda_{j-1} + \epsilon_j$\;
\STATE Calculate $A(\lambda_j)$ and $b(\lambda_j)$
with the linearization being performed at $\bar{\eta}$\;
\STATE Migrate $\bar{\eta}$:
$\bar{\eta} = \bar{\eta} + \epsilon_j(A(\lambda_j)\bar{\eta}+b(\lambda_j))$\;
\STATE Set $C=(I+\epsilon_{j} A(\lambda_j))C$\;
\vspace{1pt}
\STATE Set $D = (I+\epsilon_{j} A(\lambda_j))D + \epsilon_{j}b(\lambda_j)$\;
\ENDFOR
\end{algorithmic}
\caption{Function{$(C,D)
=F(\bar{\eta})$}}
\label{alg:flow_EDH_parameters}
\end{algorithm}

Similarly, for the PF-PF (LEDH), the invertible mapping can be
expressed as
\begin{align}
\eta_1^i = C^i \eta_0^i + D^i\,\,,
\label{eq:LEDH_mapping_new}
\end{align}
where $(C^i, D^i) = F(\bar{\eta}^i_0)$.
In this expression, $\bar{\eta}^i_0 = g_k(x_{k-1}^i,0)$ where
$g_k(\cdot,\cdot)$ is the dynamic model introduced
in Equation~\eqref{eq:dynamic}. The proposal density
becomes
\begin{align}
q(\eta_1^i|x_{k-1}^i,z_k) =\dfrac{p(\eta_0^i|x_{k-1}^i)}{|\det(C^i)|}\,\,.
\label{eq:PFPF_LEDH_proposal}
\end{align}

\subsubsection{SmHMC with LEDH}

One of the composite MH kernels we propose
uses the invertible particle flow based on the LEDH flow and
is presented in Algorithm~\ref{alg:SMCMC_LEDH}.

\begin{algorithm}[h]
\begin{algorithmic}[1]
\LCOMMENT {\underline{Joint draw of $x_{k,0:k}^i$:}}
\STATE Draw $x_{k,0:k-1}^{*(i)} \sim \widehat{\pi}_{k-1}(x_{0:k-1})$\;
\STATE Sample $\eta_0^{*(i)} = g_k(x_{k,k-1}^{*(i)},v_k)$,\\
calculate $\bar{\eta}_0^{*(i)} = g_k(x_{k,k-1}^{*(i)},0)$\;
\STATE Perform invertible particle flow (Algorithm~\ref{alg:flow_EDH_parameters})
$(C^{*(i)}, D^{*(i)}) = F(\bar{\eta}_0^{*(i)})$\;
\STATE Calculate $x_{k,k}^{*(i)} = C^{*(i)}\eta_0^{*(i)}+D^{*(i)}$\;
\STATE Compute the MH acceptance probability
$\rho_1= \min\big(1,\frac{p(x_{k,k}^{*(i)}|x_{k,k-1}^{*(i)}) p(z_k|x_{k,k}^{*(i)})
|\det(C^{*(i)})|p(\eta_0^{i-1}|x_{k,k-1}^{i-1})}
{p(\eta_0^{*(i)}|x_{k,k-1}^{*(i)})p(x_{k,k}^{i-1}|x_{k,k-1}^{i-1})
p(z_k|x_{k,k}^{i-1})|\det(C^{i-1})|}\big)$\;
\vspace{1mm}
\STATE Accept $x_{k,0:k}^{i} = x_{k,0:k}^{*(i)}$,
$\eta_0^i = \eta_0^{*(i)}$,
$C^i = C^{*(i)}$ and $D^i = D^{*(i)}$
with probability $\rho_1$.\\
Otherwise set $x_{k,0:k}^{i} = x_{k,0:k}^{i-1}$,
$\eta_0^i = \eta_0^{i-1}$,
$C^i = C^{i-1}$ and $D^i = D^{i-1}$\;
\STATE Individual refinements of $x_{k,0:k}^i$
using Algorithm~\ref{alg:SmHMC_refinement}\;
\STATE Calculate $\eta_0^{i} = ({C^i})^{-1}(x_{k,k}^{i}-D^i)$\;
\end{algorithmic}
\caption{Composite MH Kernels constructed with the manifold Hamiltonian Monte Carlo kernel and the invertible particle flow with LEDH, at the
$i$-th MCMC iteration of $k$-th time step.
\newline Input: $x_{k,0:k}^{i-1}, \eta_0^{i-1}, C^{i-1}$.
\newline Output: $x_{k,0:k}^{i}, \eta_0^{i},C^i$.}
\label{alg:SMCMC_LEDH}
\end{algorithm}

\begin{algorithm}[h]
\begin{algorithmic}[1]
\LCOMMENT {\underline{Individual refinement of $x_{k,0:k-1}^i$}:}
\STATE Draw $x_{k,0:k-1}^{*(i)} \sim \widehat{\pi}_{k-1}(x_{0:k-1})$\;
\STATE Compute the MH acceptance probability
$\rho_2 = \min\big(1,\frac{p(x_{k,k}^i|x_{k,k-1}^{*(i)})}{p(x_{k,k}^i|x_{k,k-1}^i)}\big)$\;
\STATE Accept $x_{k,0:k-1}^i = x_{k,0:k-1}^{*(i)}$ with probability $\rho_2$\;
\LCOMMENT {\underline{Individual refinement of $x_{k,k}^i$}:}
\STATE Propose $x_{k,k}^{*(i)} \sim q_{k,3}(x_k|x_{k,k-1:k}^i,z_{k})$
using the manifold Hamiltonian MCMC kernel\;
\STATE Compute the MH acceptance probability
$\rho_3 = \min\big(1,\frac{p(x_{k,k}^{*(i)}|x_{k,k-1}^i)p(z_k|x_{k,k}^{*(i)})}
{q_{k,3}(x_{k,k}^{*(i)}|x_{k,k-1:k}^i,z_{k})}
\frac{q_{k,3}(x_{k,k}^i|x_{k,k-1}^i,x_{k,k}^{*(i)},z_{k})}{p(x_{k,k}^i|x_{k,k-1}^i)p(z_k|x_{k,k}^i)}\big)$\;
\STATE Accept $x_{k,k}^i = x_{k,k}^{*(i)}$ with probability
$\rho_3$\;
\end{algorithmic}
\caption{Individual refinement steps of composite MH Kernels constructed with the manifold Hamiltonian Monte Carlo kernel, at the
$i$-th MCMC iteration of $k$-th time step.
\newline Input: $x_{k,0:k}^{i}$.
\newline Output: $x_{k,0:k}^{i}$.}
\label{alg:SmHMC_refinement}
\end{algorithm}

In the $i$-th MCMC iteration at time step $k$, we first sample
$x_{k,0:k-1}^{*(i)} \sim \widehat{\pi}_{k-1}(x_{0:k-1})$ from the approximate joint posterior distribution at time $k{-}1$.
Then, we calculate $\bar{\eta}^{*(i)}_0 = g_k(x_{k,k-1}^{*(i)},0)$ to obtain the auxiliary LEDH
flow parameters $(C^{*(i)}, D^{*(i)}) = F(\bar{\eta}^{*(i)}_0)$,
and apply the flow to the propagated particle $\eta^{*(i)}_0 = g_k(x_{k,k-1}^{*(i)},v_k)$.
Thus the proposed particle is generated as:
$x_{k,k}^i = \eta_1^{*(i)} = C^{*(i)} \eta_0^{*(i)} + D^{*(i)}$.

For this proposal, the acceptance rate of the joint draw in Algorithm~\ref{alg:SMCMC_LEDH}
can be derived using Equations~\eqref{eq:approx_posterior} and~\eqref{eq:PFPF_LEDH_proposal}:

\begin{align}
&\rho_1 =  \min\left(1,\frac{\breve \pi_k(x_{k,0:k}^{*(i)})\widehat{\pi}_{k-1}(x_{k,0:k-1}^{i-1})q(x_{k,k}^{i-1}|x_{k-1}^{i-1},z_k)}
{\widehat{\pi}_{k-1}(x_{k,0:k-1}^{*(i)})q(x_{k,k}^{*(i)}|x_{k-1}^{*(i)},z_k)\breve \pi_k(x_{k,0:k}^{i-1})}\right)\,,\nonumber\\
= & 1\wedge\frac{p(x_{k,k}^{*(i)}|x_{k,k-1}^{*(i)}) p(z_k|x_{k,k}^{*(i)})|\det(C^{*(i)})|
p(\eta_0^{i-1}|x_{k,k-1}^{i-1})}
{p(\eta_0^{*(i)}|x_{k,k-1}^{*(i)})p(x_{k,k}^{i-1}|x_{k,k-1}^{i-1})
p(z_k|x_{k,k}^{i-1})|\det(C^{i-1})|}\,\,,
\label{eq:acceptance_rate_rho_1_LEDH}
\end{align}
where we use $\wedge$ to denote the minimum operator.

When evaluating Equation~\eqref{eq:acceptance_rate_rho_1_LEDH}
in Line 5 of Algorithm~\ref{alg:SMCMC_LEDH},
the values of $x_{k,k}^{i-1}$, $\eta_0^{i-1}$ and $C^{i-1}$ are needed.
Since $x_{k,k}^{i-1}$ may be generated by the manifold Hamiltonian Monte Carlo kernel
$q_{k,3}(\cdot)$ as shown in Algorithm~\ref{alg:SmHMC_refinement}, the corresponding
$\eta_0^{i-1}$ is not available through Lines 2 and 6
of Algorithm~\ref{alg:SMCMC_LEDH}.
This can be resolved using the invertible mapping property
of the invertible particle flow.
As $C^{i-1}$ is invertible, we can calculate
$\eta_0^{i-1}$ given $x_{k,k}^{i-1}$
by solving Equation~\eqref{eq:LEDH_mapping_new}:
\begin{align}
\eta_0^{i-1} = (C^{i-1})^{-1}(x_{k,k}^{i-1}-D^{i-1})\,\,.
\end{align}

\subsubsection{SmHMC with EDH}

Calculation of individual flow parameters
at every MCMC iteration in Algorithm~\ref{alg:SMCMC_LEDH} can be computationally expensive.
Similar to the spirit of the PF-PF (EDH) \cite{li2017},
we can calculate the flow parameters $C$ and $D$ only once,
using an auxiliary state variable derived from the samples, and apply the calculated
flow parameters for all MCMC iterations.
The resulting procedure is described in Algorithm \ref{alg:SmHMC_EDH}.
The flow parameters $C$ and $D$ are calculated only once in the initialization of each time step $k$, as in Algorithm~\ref{alg:SmHMC_EDH_flow_param}.

\begin{algorithm}[h]
\begin{algorithmic}[1]
\LCOMMENT {\underline{Joint draw of $x_{k,0:k}^i$:}}
\STATE Draw $x_{k,0:k-1}^{*(i)} \sim \widehat{\pi}_{k-1}(x_{0:k-1})$\;
\STATE Sample $\eta_0^{*(i)} = g_k(x_{k,k-1}^{*(i)},v_k)$\;
\STATE Calculate $x_{k,k}^{*(i)} = C\eta_0^{*(i)}+D$\;
\STATE Compute the MH acceptance probability
$\rho_1= \min\big(1,\frac{p(x_{k,k}^{*(i)}|x_{k,k-1}^{*(i)}) p(z_k|x_{k,k}^{*(i)})
p(\eta_0^{i-1}|x_{k,k-1}^{i-1})}
{p(\eta_0^{*(i)}|x_{k,k-1}^{*(i)})p(x_{k,k}^{i-1}|x_{k,k-1}^{i-1})
p(z_k|x_{k,k}^{i-1})}\big)$\;
\vspace{1mm}
\STATE Accept $x_{k,0:k}^{i} = x_{k,0:k}^{*(i)}$,
$\eta_0^i = \eta_0^{*(i)}$ with probability $\rho_1$.\\
Otherwise set $x_{k,0:k}^{i} = x_{k,0:k}^{i-1}$,
$\eta_0^i = \eta_0^{i-1}$\;
\STATE Individual refinements of $x_{k,0:k}^i$
using Algorithm~\ref{alg:SmHMC_refinement}\;
\STATE Calculate $\eta_0^{i} = C^{-1}(x_{k,k}^{i}-D)$\;
\end{algorithmic}
\caption{Composite MH Kernels constructed with the manifold Hamiltonian Monte Carlo kernel and the invertible particle flow with EDH, at the
$i$-th MCMC iteration of $k$-th time step.
$C$ and $D$ were pre-computed using Algorithm~\ref{alg:SmHMC_EDH_flow_param}.
\newline Input: $x_{k,0:k}^{i-1}, \eta_0^{i-1}, C, D$.
\newline Output: $x_{k,0:k}^{i}, \eta_0^{i}$.}\label{alg:SmHMC_EDH}
\end{algorithm}

\begin{algorithm}[h]
\begin{algorithmic}[1]
\STATE Draw $x_{k,0:k-1}^{0} \sim \widehat{\pi}_{k-1}(x_{0:k-1})$\;
\STATE Sample $\eta_0^0 = g_k(x_{k,k-1}^0,v_k)$,\\
calculate $\bar{\eta}_0 = g_k(\bar{x}_{k-1,k-1},0)$\;
\STATE Perform invertible particle flow (Algorithm~\ref{alg:flow_EDH_parameters})
$(C, D) = F(\bar{\eta}_0)$\;
\end{algorithmic}
\caption{The flow parameter calculation for SmHMC (EDH).}
\label{alg:SmHMC_EDH_flow_param}
\end{algorithm}

The calculation of the acceptance rate in the joint draw step
can be further simplified compared to Equation~\eqref{eq:acceptance_rate_rho_1_LEDH}
as the same mapping of the flow is applied to each particle.
The candidate particle $x_{k,k}^{*(i)}$ and the particle
$x_{k,k}^{i-1}$ share the same value of $C$ in their proposal densities
(Equation~\eqref{eq:PFPF_EDH_proposal}).
Thus, for the SmHMC algorithm with the EDH flow, we have
\begin{align}
\rho_1 = \min\left(1,\frac{p(x_{k,k}^{*(i)}|x_{k,k-1}^{*(i)}) p(z_k|x_{k,k}^{*(i)})
p(\eta_0^{i-1}|x_{k,k-1}^{i-1})}
{p(\eta_0^{*(i)}|x_{k,k-1}^{*(i)})p(x_{k,k}^{i-1}|x_{k,k-1}^{i-1})
p(z_k|x_{k,k}^{i-1})}\right)\,\,.
\label{eq:acceptance_rate_rho_1_EDH}
\end{align}

\subsection{SmHMC with LEDH for GMM distributed noises}
When the process and measurement noises are distributed as Gaussian
mixtures, the posterior distribution $\pi_k(x_{0:k})$ becomes
multi-modal. In order to explore different modes of the posterior
efficiently, we consider an extended state-space $(x_k,d_k,c_k)$, as
in \cite{pal2018}, and define the joint posterior as
\begin{align}
&\pi_k(x_{0:k},d_{1:k},c_{1:k}) =  p(x_{0:k},d_{1:k},c_{1:k}|z_{1:k}) \,,\nonumber\\
&\quad \quad \propto  p(d_k)p(c_k)p(x_k|x_{k-1},d_k)p(z_k|x_k,c_k) \times\nonumber\\
&\quad \quad \quad \quad \quad \pi_{k-1}(x_{0:k-1},d_{1:k-1},c_{1:k-1})\,,
\label{eq:recursive_joint_gmm} 
\end{align}
 which admits $\pi_k(x_{0:k})$ as its $x_{0:k}$ marginal. Similar to Equation \eqref{eq:approx_posterior}, 
 based on the approximate joint posterior $\widehat{\pi}_{k-1}(x_{0:k-1},d_{1:k-1},c_{1:k-1})$, of the previous time step $k-1$, we approximate $\pi_k(x_{0:k},d_{1:k},c_{1:k})$ as follows:
\begin{align}
\breve \pi(x_{0:k},d_{1:k},c_{1:k}) &\propto  p(d_k)p(c_k)p(x_k|x_{k-1},d_k)p(z_k|x_k,c_k)\nonumber\\
&\quad \quad \times \widehat \pi_{k-1}(x_{0:k-1},d_{1:k-1},c_{1:k-1})\,.
\label{eq:approx_posterior_gmm} 
\end{align} 

For this model, in the joint draw step of SMCMC, we adopt the following strategy. First we sample
\begin{align}
(x_{k,0:k-1}^{*(i)},d_{k,1:k-1}^{*(i)},c_{k,1:k-1}^{*(i)}) \sim \widehat \pi_{k-1}(x_{0:k-1},d_{1:k-1},c_{1:k-1})\,.
\end{align}
Similar to the spirit of auxiliary particle filtering \cite{pitt1999}, we design efficient measurement-driven proposals for sampling $(d_k,c_k)$, 
in contrast to \cite{pal2018}, where the switching variables are sampled from their respective priors. To sample $d_{k,k}^{*(i)}$, we use 
the following proposal:
\begin{align}
q(d_{k,k}^{*(i)} &\!=\! m| x_{k,k-1}^{*(i)},z_{k})  \propto  p(d_{k,k}^{*(i)}\!=\!m)\nonumber\\
&\quad \quad \times p(\bar{\eta}_{0}^{*(i)}|x_{k,k-1}^{*(i)},d_{k,k}^{*(i)}\!=\!m)p(z_k|\bar{\eta}_{0}^{*(i)})\,,
\label{eq:proposal_d_k}
\end{align}
where $\bar{\eta}^{*(i)}_{0} = g_k(x_{k,k-1}^{*(i)},\psi_{k,m})$. Conditioned on $d_{k,k}^{*(i)}\!=\!m$, we sample 
$c_{k,k}^{*(i)}$ from
\begin{align}
q(c_{k,k}^{*(i)}\!=\!n|x_{k,k-1}^{*(i)},d_{k,k}^{*(i)}\!=\!m,&\,z_{k})  \propto p(c_{k,k}^{*(i)}\!=\!n)\nonumber\\
&\times p(z_k|\bar{\eta}_{0}^{*(i)},c_{k,k}^{*(i)}\!=\!n)\,.
\label{eq:proposal_c_k}
\end{align}
Then conditioned on $(d_{k,k}^{*(i)}=m,c_{k,k}^{*(i)}=n)$, we calculate $\bar{\eta}^{*(i)}_{0} = g_k(x_{k,k-1}^{*(i)},\psi_{k,m})$ to obtain the auxiliary LEDH
flow parameters $(C^{*(i)}, D^{*(i)}) = F(\bar{\eta}^{*(i)}_{0})$, using the $m$-th and $n$-th 
Gaussian component of the dynamic and measurement model respectively.
Then this flow is applied to the propagated particle $\eta^{*(i)}_{0} = g_k(x_{k,k-1}^{*(i)},v_{k,m})$, where $v_{k,m} \sim \mathcal{N}(\psi_{k,m}, Q_{k,m})$ is the $m$-th component 
in the process noise. The proposed particle is generated as:
$x_{k,k}^i = \eta_1^{*(i)} = C^{*(i)} \eta_{0}^{*(i)} +
D^{*(i)}$. Using the invertible mapping property, established by the flow,
we can calculate
\begin{align}
q(x_{k,k}^{*(i)}|x_{k,k-1}^{*(i)},&\,d_{k,k}^{*(i)}\!=\!m,c_{k,k}^{*(i)}\!=\!n,z_k)\nonumber\\
&= \dfrac{p(\bar{\eta}_{0}^{*(i)}|x_{k,k-1}^{*(i)},d_{k,k}^{*(i)}\!=\!m)}{|\det(C^{*(i)})|}\,.
\label{eq:GMM_LEDH_proposal}
\end{align}
Using Equations~\eqref{eq:approx_posterior_gmm}, \eqref{eq:proposal_d_k}, \eqref{eq:proposal_c_k} and \eqref{eq:GMM_LEDH_proposal}, 
the acceptance rate for the joint draw of $(x_{k,0:k}^{*(i)},d_{k,1:k}^{*(i)}c_{k,1:k}^{*(i)})$ using the proposed kernel can be calculated as:
\begin{align}
\rho_1= &\min\left(1,\tfrac{p(d_{k,k}^{*(i)})p(c_{k,k}^{*(i)})p(x_{k,k}^{*(i)}|x_{k,k-1}^{*(i)},d_{k,k}^{*(i)})}
{q(d_{k,k}^{*(i)}|x_{k,k-1}^{*(i)},z_{k})q(c_{k,k}^{*(i)}|x_{k,k-1}^{*(i)},d_{k,k}^{*(i)},z_{k})}\right.\nonumber\\
&\times\tfrac{q(d_{k,k}^{i-1}|x_{k,k-1}^{i-1},z_{k})q(c_{k,k}^{i-1}|x_{k,k-1}^{i-1},d_{k,k}^{i-1},z_{k})}{p(d_{k,k}^{i-1})p(c_{k,k}^{i-1})p(x_{k,k}^{i-1}|x_{k,k-1}^{i-1},d_{k,k}^{i-1})}\nonumber\\
&\left. \times\tfrac{|\det(C^{*(i)})|p(\bar{\eta}_{0}^{i-1}|x_{k,k-1}^{i-1},d_{k,k}^{i-1})p(z_k|x_{k,k}^{*(i)},c_{k,k}^{*(i)})}{p(\bar{\eta}_{0}^{*(i)}|x_{k,k-1}^{*(i)},d_{k,k}^{*(i)})|\det(C^{i-1})|p(z_k|x_{k,k}^{i-1},c_{k,k}^{i-1})}\right).
\label{eq:acceptance_rate_rho_1_GMM}
\end{align}
For individual refinement of $x_{k,0:k-1}^i$, we use the independent proposal $q_{k,2} = \widehat{\pi}_{k-1}$. 
We can compute the acceptance rate of the refinement as follows:
\begin{align}
\rho_2 &= \min\left(1,\tfrac{\breve \pi_k(x_{k,0:k-1}^{*(i)},x_{k,k}^{i},d_{k,1:k}^{i},c_{k,1:k}^{i})}
{\widehat{\pi}_{k-1}(x_{k,0:k-1}^{*(i)})}\tfrac{\widehat{\pi}_{k-1}(x_{k,0:k-1}^{i})}{\breve \pi_k(x_{k,0:k}^{i},d_{k,1:k}^{i},c_{k,1:k}^{i})}\right)\nonumber\\
& = \min\left(1,\tfrac{p(x_{k,k}^i|x_{k,k-1}^{*(i)},d_{k,k}^{i})}{p(x_{k,k}^i|x_{k,k-1}^i,d_{k,k}^{i})}\right)\,.
\label{eq:rho_2_calculation_gmm}
\end{align}

\begin{algorithm}[htbp]
\begin{algorithmic}[1]
\LCOMMENT {\underline{Joint draw of $x_{k,0:k}^i$:}}
\STATE Draw $x_{k,0:k-1}^{*(i)} \sim \widehat{\pi}_{k-1}(x_{0:k-1})$\;
\STATE Sample $d_{k,k}^{*(i)}\!=\! m \in \{1,2,\cdots M\}$ from \\$q(d_k|x_{k,k-1}^{*(i)},z_{k})$\;
\STATE Sample $\eta_0^{*(i)} = g_k(x_{k,k-1}^{*(i)},v_{k,m})$, where $v_{k,m} \sim \mathcal{N}(\psi_{k,m},Q_{k,m})$.
calculate $\bar{\eta}_0^{*(i)} = g_k(x_{k,k-1}^{*(i)},\psi_{k,m})$\;
\STATE Sample $c_{k,k}^{*(i)}\!=\! n \in \{1,2,\cdots N\}$ from $q(c_k|x_{k,k-1}^{*(i)},d_{k,k}^{*(i)}\!=\!m,z_{k})$\;
\STATE Perform invertible particle flow (Algorithm~\ref{alg:flow_EDH_parameters})
$(C^{*(i)}, D^{*(i)}) = F(\bar{\eta}_0^{*(i)})$ using $m$-th and $n$-th component of dynamic and measurement models respectively\;
\STATE Calculate $x_{k,k}^{*(i)} = C^{*(i)}\eta_0^{*(i)}+D^{*(i)}$\;
\STATE Compute the MH acceptance probability
\small
$\rho_1= \min\big(1,\dfrac{p(d_{k,k}^{*(i)})p(c_{k,k}^{*(i)})p(x_{k,k}^{*(i)}|x_{k,k-1}^{*(i)},d_{k,k}^{*(i)})}
{q(d_{k,k}^{*(i)}|x_{k,k-1}^{*(i)},z_{k})q(c_{k,k}^{*(i)}|x_{k,k-1}^{*(i)},d_{k,k}^{*(i)},z_{k})}$\\
$\dfrac{q(d_{k,k}^{i-1}|x_{k,k-1}^{i-1},z_{k})q(c_{k,k}^{i-1}|x_{k,k-1}^{i-1},d_{k,k}^{i-1},z_{k})}{p(d_{k,k}^{i-1})p(c_{k,k}^{i-1})p(x_{k,k}^{i-1}|x_{k,k-1}^{i-1},d_{k,k}^{i-1})}$\\
$\dfrac{|det(C^{*(i)})|p(\bar{\eta}_{0}^{i-1}|x_{k,k-1}^{i-1},d_{k,k}^{i-1})p(z_k|x_{k,k}^{*(i)},c_{k,k}^{*(i)})}{p(\bar{\eta}_{0}^{*(i)}|x_{k,k-1}^{*(i)},d_{k,k}^{*(i)})|det(C^{i-1})|p(z_k|x_{k,k}^{i-1},c_{k,k}^{i-1})}\big)$\;
\normalsize
\STATE Accept $x_{k,0:k}^{i} = x_{k,0:k}^{*(i)}$,
$\eta_0^i = \eta_0^{*(i)}$, $d_{k,k}^i = d_{k,k}^{*(i)}$, $c_{k,k}^i = c_{k,k}^{*(i)}$,
$C^i = C^{*(i)}$ and $D^i = D^{*(i)}$
with probability $\rho_1$.\\
Otherwise set $x_{k,0:k}^{i} = x_{k,0:k}^{i-1}$,
$\eta_0^i = \eta_0^{i-1}$, $d_{k,k}^i = d_{k,k}^{i-1}$, $c_{k,k}^i = c_{k,k}^{i-1}$,
$C^i = C^{i-1}$ and $D^i = D^{i-1}$\;
\STATE Individual refinements of $x_{k,0:k}^i$
using Algorithm~\ref{alg:SmHMC_refinement}
given $d_{k,k}^i$ and $c_{k,k}^i$\;
\STATE Calculate $\eta_0^{i} = ({C^i})^{-1}(x_{k,k}^{i}-D^i)$\;
\end{algorithmic}
\caption{Composite MH Kernels for models with Gaussian mixture noises, constructed with the manifold Hamiltonian Monte Carlo kernel and the invertible particle flow with LEDH, at the
$i$-th MCMC iteration of $k$-th time step.
\newline Input: $x_{k,0:k}^{i-1}, \eta_0^{i-1}, d_{k,k}^{i-1},c_{k,k}^{i-1}, C^{i-1}$.
\newline Output: $x_{k,0:k}^{i}, \eta_0^{i}, d_{k,k}^{i},c_{k,k}^{i}, C^i$.}
\label{alg:SmHMC_LEDH_GMM}
\end{algorithm}

The algorithm is summarized in Algorithm \ref{alg:SmHMC_LEDH_GMM}. From Equation \eqref{eq:acceptance_rate_rho_1_GMM},
we note that we can discard $\{(d_{k,k}^{j},c_{k,k}^{j})\}_{j=N_b+1}^{N_b+N_p}$ after every time step $k$, 
if we are only interested in approximating $\pi_k(x_{0:k})$.

\subsection{Convergence results}
\label{sec:convergence}

In this section, we present some theoretical results regarding
convergence of the approximate joint posterior distribution
$\widehat{\pi}_k(x_{0:k})$ to $\pi_k(x_{0:k})$ for every $k \geq 0$.
As in Section~\ref{sec:related_work}, we use $\pi_k(x_{0:k})$ to
denote our target distribution $p(x_{0:k}|z_{1:k})$. We initialize by
setting $\pi_0(x_{0}) = p(x_0)$.  To facilitate concise presentation
of the results, we use a simplified notation in
this subsection and Appendices~\ref{appendix:proof_theo1} and ~\ref{appendix:proof_corr}, where $x_{0:k}^{i}$ and
$x_{0:k}^{*(i)}$ denote $x_{k,0:k}^{i}$ and $x_{k,0:k}^{*(i)}$,
respectively. We also use $\widehat{\pi}_k^{(N_p)}$ and
$\breve{\pi}_k^{(N_p)}$ to denote $\widehat{\pi}_k$ and
$\breve{\pi}_k$ to indicate explicitly that they are comprised of $N_p$ MCMC samples. While
proving the theorems, we assume that the burn-in period $N_b=0$, for
simplicity. The results are, however valid for any non-zero $N_b$. Our
results are derived for Algorithms~\ref{alg:vanilla_SMCMC},
\ref{alg:composite_kernel}, \ref{alg:SMCMC_LEDH}
and~\ref{alg:SmHMC_EDH}. However, they can be easily extended to apply
to Algorithm \ref{alg:SmHMC_LEDH_GMM} as well, by considering convergence in the extended
state space $(x_{0:k},d_{1:k},c_{1:k})$ which implies convergence of $x_{0:k}$.  

We assume that $\pi_k(x_{0:k})$ is defined on a
measurable space $(E_k,\mathcal{F}_k)$ where $E_0 = E$, $\mathcal{F}_0
= \mathcal{F}$, and $E_k = E_{k-1}\times E$, $\mathcal{F}_k =
\mathcal{F}_{k-1}\times \mathcal{F}$. We denote by $\mathcal{P}(E_k)$
the set of probability measures on $(E_k,\mathcal{F}_k)$.  We also
define $\mathcal{S}_k = \{x_{0:k} \in E_k: \pi_k(x_{0:k}) > 0\}$. For
$k \geq 0$, $\pi_k$ is known up to a normalizing constant $0 < Z_k < \infty$. We have
\begin{align}
\pi_k(x_{0:k}) = \dfrac{p(x_0)\displaystyle\prod_{l=1}^k p(x_l|x_{l-1})p(z_l|x_l)}{p(z_{1:k})} = \dfrac{\gamma_k(x_{0:k})}{Z_k}\,,
\label{eq:pi_gamma}
\end{align}
where $\gamma_k: E_k \rightarrow \mathbb{R}^{+}$ is known
point-wise and the normalizing constant $Z_k = \int_{E_k} \gamma_k(dx_{0:k}) = p(z_{1:k})$ is the unknown marginal likelihood of the observations.
For the joint draw step in SMCMC, the proposal distribution at time
$k=0$ is $q_0(x_0)$, and for $k > 0$ it is $q_k(x_{0:k-1},x_k)$, where $q_k: \mathbb{E}_{k-1} \times \mathbb{E} \rightarrow \mathbb{R}^{+}$ is a probability
density in its last argument $x_k$ conditional on
its previous arguments $x_{0:k-1}$. For $k>0$ and for any measure $\mu_{k-1} \in \mathcal{P}(E_{k-1})$, we define
\begin{align}
(\mu_{k-1}\times q_k)(d x_{0:k}) = \mu_{k-1}(d x_{0:k-1})q_k(x_{0:k-1},d x_k)\,.
\label{eq:mu_q_def}
\end{align}

Based on the proposal distributions $q_0(x_0)$ and $q_k(x_{0:k-1},x_k)$, we define importance weights:
 \begin{align}
 w_0(x_{0}) = \dfrac{\gamma_0(x_{0})}{q_0(x_{0})}\,;
 \label{eq:w_0_def}
 \end{align}
 and for $k > 0$
\begin{align}
w_k(x_{0:k}) = \dfrac{\gamma_k(x_{0:k})}{\gamma_{k-1}(x_{0:k-1})q_k(x_{0:k-1},x_k)}\,.
\label{eq:w_k_def}
\end{align}
Combining Equations~\eqref{eq:pi_gamma}, \eqref{eq:mu_q_def}, \eqref{eq:w_0_def} and~\eqref{eq:w_k_def},
it is easy to derive that
\begin{align}
Z_0 = \int_{E_0} w_0(x_{0}) q_0(d x_{0})\,,
\label{def:Z_0}
\end{align}
and for $k > 0$
\begin{align}
\dfrac{Z_k}{Z_{k-1}} = \int_{E_k} w_k(x_{0:k})(\pi_{k-1}\times q_k)(d x_{0:k})\,.
\label{def:Z_k_ratio}
\end{align}
Asymptotically, $x_{0:k}^{*(i)}$ is distributed according to
$(\pi_{k-1}\times q_k)(x_{0:k})$. Thus, the (ratio of)
normalizing constants can easily be estimated as
\begin{align}
\reallywidehat{Z_{0}}^{(N_p)} = \dfrac{1}{N_p}\sum_{i=1}^{N_p}w_0(x_{0}^{*(i)})\,,
\label{eq:Z_0}
\end{align}
and for $k > 0$
\begin{align}
\reallywidehat{\left(\dfrac{Z_k}{Z_{k-1}}\right)}^{(N_p)} = \dfrac{1}{N_p}\sum_{i=1}^{N_p}w_k(x_{0:k}^{*(i)})\,.
\label{eq:Z_k_ratio}
\end{align}

We use $\widehat{\pi}_{k}^{(N_p)}(x_{0:k})$ to denote the empirical measure approximation of the target
distribution $\pi_{k}(x_{0:k})$.
\begin{align}
\widehat{\pi}_{k}^{(N_p)}(x_{0:k}) = \frac{1}{N_p}\sum_{i=1}^{N_p}\delta_{x_{0:k}^{i}}(x_{0:k})\,.
\end{align}
For any measure $\mu$ and integrable test function $f:E \rightarrow
\mathbb{R}$, we define $\mu(f) = \int_E f(x) \mu(dx)$. We denote that
$\mathcal{L}^p(E_k,\mathcal{F}_k,\mu_k) = \{f_k: E_k \rightarrow
\mathbb{R}$ such that $f_k$ is measurable with respect to $\mathcal{F}_k$ and
$\mu_k(|f_k|^p) < \infty\}$ for $p \geq 1$. In other words,
$\mathcal{L}^p(E_k,\mathcal{F}_k,\mu_k)$ is the set of real-valued,
$\mathcal{F}_k$-measurable functions defined on $E_k$, whose absolute
$p$'th moment exists with respect to $\mu_k$. 

The following theorem requires the following relatively weak
assumption to be satisfied: 
\begin{assumption}
For any time index $k \ge 0$, there exists $B_k < \infty$ such that for any $x_{0:k} \in \mathcal{S}_k$,
we have $w_k(x_{0:k}) \leq B_k$.
\label{assump:bound}
\end{assumption}

Our first result establishes almost sure convergence of the
approximating distribution to the target distribution. See Appendix~\ref{appendix:proof_theo1} for the proof of the theorem.
\begin{theorem}
Given Assumption~\ref{assump:bound}, for any $k \ge 0$,
$x_{0:l}^{(0)} \in \mathcal{S}_l$ for $0 \leq l \leq k$ and $f_k \in \mathcal{L}^1(E_k,\mathcal{F}_k,\pi_k)$,
\begin{align}
\widehat{\pi}_{k}^{(N_p)}(f_k) \longrightarrow \pi_k(f_k)\,
\end{align} 
almost surely, as $N_p \rightarrow \infty$.
\label{theorem:empirical_avg}
\end{theorem}

With Equations~\eqref{eq:Z_0} and~\eqref{eq:Z_k_ratio}, a
straightforward corollary of Theorem~\ref{theorem:empirical_avg}
follows. See Appendix~\ref{appendix:proof_corr} for the proof.
\begin{corollary}
Given Assumption~\ref{assump:bound}, for any $k \ge 0$ and $x_{0:l}^{(0)} \in \mathcal{S}_l$ for $0 \leq l \leq k$,
\begin{align}
\reallywidehat{Z_0}^{(N_p)} \longrightarrow Z_0\,,
\end{align}
and for $k > 0$,
\begin{align}
\reallywidehat{\left(\dfrac{Z_k}{Z_{k-1}}\right)}^{(N_p)} \longrightarrow \dfrac{Z_k}{Z_{k-1}}
\,,\label{corr:Z_k_ratio}
\end{align}
almost surely, as $N_p \rightarrow \infty$.
\label{cor:Z_k_ratio}
\end{corollary}

\section{Simulation and Results}
\label{sec:result}

In this section we describe the numerical simulation experiments that we have conducted to evaluate the performance of the proposed algorithms. The experiments explore idealized scenarios in order to focus on the sampling capabilities of the compared algorithms. More extensive experimentation is required to explore the impact of practical issues in sensor networks.

\subsection{Large spatial sensor networks example}

We first examine the proposed algorithms in simulations with setups of large spatial sensor networks.
$d$ sensors are evenly deployed on a two-dimensional grid, at coordinates
$\{1,2,\ldots,\sqrt{d}\} \times \{1,2,\ldots,\sqrt{d}\}$. Each sensor collects measurements
independently with respect to the other sensors, about the underlying states at the sensor's location.
Such networks can be useful in environment monitoring, weather forecasts and surveillance.  

The SmHMC algorithm leads to the smallest average mean squared error in these large spatial sensor network examples~\cite{septier2016}, hence we would like to evaluate the proposed algorithms in the same examples.

The full state at time step $k$ is denoted by $x_k = [x_k^1, \ldots, x_k^d]$,
where $x_{k}^{s} \in \mathbb{R}$ is the state at the $s${-}th sensor's position for $s \in \{1,\ldots,d\}$.
$x_k$ evolves according to a multivariate Generalized Hyperbolic (GH) skewed-$t$ distribution, which
is a heavy-tailed distribution useful for modeling physical processes, extreme events and financial markets~\cite{zhu2010}:
\begin{align}
p(x_k|x_{k-1}) =&\dfrac{e^{(x_k-\alpha x_{k-1})^T\Sigma^{-1}\gamma}}{\sqrt{(\nu+Q(x_k))(\gamma^T\Sigma^{-1}\gamma)}^{-\frac{\nu+d}{2}}(1+\frac{Q(x_k)}{\nu})^{\frac{\nu+d}{2}}}\nonumber\\
& \times K_{\frac{\nu+d}{2}}(\sqrt{(\nu+Q(x_k))(\gamma^T\Sigma^{-1}\gamma)})\,\,.
\label{eq:mGH}
\end{align}
Here $\alpha$ is a scalar, the parameters $\gamma$ and $\nu$ determine the shape of the distribution, $K_{\frac{\nu+d}{2}}$ is the modified Bessel function 
of the second kind of order $\frac{\nu+d}{2}$, and $Q(x_k) = (x_k-\alpha x_{k-1})^T\Sigma^{-1}(x_k-\alpha x_{k-1})$.
The $(i,j)$-th entry of the dispersion matrix $\Sigma$ is defined by:
\begin{align}
\Sigma_{i,j} = \alpha_0 e^{-\frac{||L^i-L^j||_2^2}{\beta}}+\alpha_1\delta_{i,j}\,.
\label{eq:dispersion}
\end{align}
We use $||\cdot||_2$ to denote the L2-norm. $L^i$ is the physical location of the $i$-th sensor, and $\delta_{i,j}$
is the Kronecker symbol.

The measurements are Poisson-distributed count data:
\begin{align}
p(z_k|x_k) = \prod_{s=1}^{d}\mathcal{P}(z_k^s;m_1 e^{m_2 x_k^s}) \,\,.
\end{align}
$\mathcal{P}(\cdot;\Lambda)$ denotes the Poisson$(\Lambda)$ distribution. $m_1$ and $m_2$ are scalars
which control the mean of the Poisson distribution.

We set $\alpha = 0.9, \alpha_0 = 3, \alpha_1 = 0.01, \beta = 20, \nu = 7$, $m_1 = 1$ and $m_2 = \frac{1}{3}$.
All elements of the vector $\gamma$ are set to $0.3$.
$d$ is set to $144$ or $400$ to represent two high-dimensional filtering examples.
Each simulation example is executed for 100 times and each simulation lasts for 10 time steps. Execution time in this paper is produced with an Intel i7-4770K 3.50GHz CPU
and 32GB RAM.
Where not specified otherwise, the number of samples $N_p$ is 200, and the burn-in period is $20$.

Table~\ref{tab:spatial} reports the average estimation errors, acceptance rates (if applicable) and execution time. We observe that by combining SmHMC with the EDH flow (SmHMC (EDH)) or the LEDH flow (SmHMC (LEDH)), the average MSE decreases compared with the vanilla SmHMC algorithm where the independent MH kernel based on prior as proposal is employed in the joint draw step.  The decrease in the MSE is due to the increase of the acceptance ratio in the joint draw step where particle flow is employed. So, samples are more diversified after the joint draw step.  The SmHMC (EDH) and SmHMC (LEDH) algorithms lead to similar average MSEs. The SmHMC (EDH) algorithm is hence preferred in this setting as it is computationally much more efficient than the SmHMC (LEDH) method. The SmHMC (LEDH) method adds negligible computational overhead compared to the SmHMC.

The EDH and the LEDH methods produce similar MSEs, again indicating that performing separate flow parameters for each particle does not offer additional gain in this setting. Although the MSEs from the EDH and the LEDH are the smallest among compared algorithms, these two algorithms are not statistically consistent. The PF-PF based on
the EDH flow has a higher MSE than SmHMC (EDH) with the same number of particles as although the
EDH flow moves particles to region where posterior densities are relatively high, the particle filter can still suffer from weight degeneracy in such high-dimensional spaces. An increased number of
particles lead to improved performance for PF-PF (EDH), but this
requires a much higher consumption of memory due to flow operations
of a large number particles.
The bootstrap particle filter (BPF) produces high average MSE even with 1 million particles. The EKF and the UKF frequently lead
to lost tracks as the posterior distributions
are strongly non-Gaussian.

\begin{table}[htbp]
\small
\setlength{\tabcolsep}{1pt}
\centering
\caption{Average MSE, acceptance rates (if applicable) and  execution time per step in the large spatial sensor
network example with a skewed - t dynamic model and count measurements. The parenthesis after the average MSE values indicates the number of lost
tracks out of 100 simulation trials where lost tracks are defined as those whose
average estimation errors are greater Than $\frac{\sqrt{d}}{2}$. The average MSE is calculated with
the simulation trials where tracking is not lost.}
\label{tab:spatial}
{
\begin{tabular}{|c|c|c|c|c|c|c|c|}
\hline
\multirow{2}{*}{$d$} & \multirow{2}{*}{Algorithm} &  \multirow{2}{*}{\begin{tabular}[c]{@{}c@{}} No. of\\ particles\end{tabular}} &
\multirow{2}{*}{\begin{tabular}[c]{@{}c@{}}Avg.\\ MSE\end{tabular}} & \multicolumn{3}{|c|}{Acceptance rate} & \multirow{2}{*}{\begin{tabular}[c]{@{}c@{}}Exec.\\ time (s)\end{tabular}}\\\cline{5-7}
& & & & $\rho_1$ & $\rho_2$ & $\rho_3$ &\\\hline
\multirow{10}{*}{144} & SmHMC (EDH) & 200 & 0.75 & 0.04 & 0.01 & 0.71 & 11.5\\\cline{2-8}
& SmHMC (LEDH) & 200 & 0.76 & 0.05 & 0.01 & 0.71 & 19\\\cline{2-8}
& SmHMC & 200 & 0.82 & 0.003 & 0.01 & 0.73 & 11\\\cline{2-8}
& EDH & 200 & 0.68 & - & - & - & 0.05\\\cline{2-8}
& EDH & 10000 & 0.68 & - & - & - & 0.6 \\\cline{2-8}
& LEDH & 200 & 0.71 & - & - & - & 7\\\cline{2-8}
& PF-PF (EDH) & 200 & 0.88 & - & - & - & 0.05\\\cline{2-8}
& PF-PF (EDH) & $10^5$ & 0.73 & - & - & - & 6.0\\\cline{2-8}
& EKF & N/A & 2.5 (28) & - & - & - & 0.002\\\cline{2-8}
& UKF & N/A & 2.4 (34) & - & - & - & 0.05\\\cline{2-8}
& BPF & $10^6$ & 1.4 (1) & - & - & - & 6.8\\\hline\hline
\multirow{10}{*}{400} & SmHMC (EDH) & 200 & 0.70 & 0.02 & 0.02 & 0.61 & 90\\\cline{2-8}
& SmHMC (LEDH) & 200 & 0.71 & 0.02 & 0.02 & 0.60 & 205 \\\cline{2-8}
& SmHMC & 200 & 0.73 & 0.002 & 0.02 & 0.63 & 88\\\cline{2-8}
& EDH & 200 & 0.60 & - & - & - & 0.5\\\cline{2-8}
& EDH & 10000 & 0.60 & - & - & - & 2.5\\\cline{2-8}
& LEDH & 200 & 0.62 & - & - & - & 88\\\cline{2-8}
& PF-PF (EDH) & 200 & 0.92 & - & - & - & 0.6\\\cline{2-8}
& PF-PF (EDH) & $10^5$ & 0.75 & - & - & - & 21\\\cline{2-8}
& EKF & N/A & 3.4 (18) & - & - & - & 0.03\\\cline{2-8}
& UKF & N/A &3.8 (27) & - & - & - & 1.2\\\cline{2-8}
& BPF & $10^6$ & 3.3 (1) & - & - & - & 23\\\hline
\end{tabular}}
\end{table}

%
\subsection{Estimation of normalizing constants from a linear Gaussian example}
\label{sec:setup_linear_gaussian}

In many Bayesian inference problems, it is very important to estimate the normalizing constant $Z_k$. For a general HMM, 
analytical evaluation of \eqref{def:Z_0} and \eqref{def:Z_k_ratio}
is not possible, as $\pi_{k-1}$ is not tractable. However, in a linear Gaussian filtering problem, the posterior 
distribution can be analytically computed from a Kalman filter, which allows exact estimation of $Z_k$. 
Here we consider a linear Gaussian setup to allow us to compare the estimated normalizing constants against the true values. The dynamic and measurement models are:
\begin{align}
x_k &=\alpha x_{k-1} + v_k\,,\\
z_k &= x_k +w_k\,,
\end{align}
where $x_k \in \mathbb{R}^d$ and $z_k\in \mathbb{R}^d$ are the state and measurement vectors respectively.  We set
$\alpha =0.9$ as in the previous section. The process noise $v_k \sim \mathcal{N}(\mathbf{0}, \Sigma)$ where $\Sigma$ is 
given in \eqref{eq:dispersion}. The measurement noise $w_k \sim \mathcal{N}(\mathbf{0}, \sigma_z^2I)$. 
We set $\sigma_z=0.5$. The true state starts at
$x_0 = \mathbf{0}$. We set $d=64$. The experiment is executed 100 times for $T=10$ time steps.

The main objective of this experiment is to compare SMCMC algorithms with different particle filters for estimation of normalizing constants.
We define the relative MSE in estimating $\log Z_k$ as follows:
\begin{align}
MSE_{\log Z}^{(rel)} = \dfrac{\sum_{k=1}^T(\log Z_k -\log \reallywidehat{Z}_k)^2}{\sum_{k=1}^T (\log Z_k)^2}\,.\nonumber
\end{align}
For this linear Gaussian setup, the Kalman filter (KF), shows the least error in state estimation in Table~\ref{tab:linear}.
It also can compute $\log Z_k$ analytically. The average acceptance rates of the SmHMC (LEDH) and SmHMC (EDH) are higher than the vanilla SmHMC, showing that the 
incorporation of particle flow in SmHMC provides many more initializations for mHMC based refinement and better mixing. Although the average error in the state estimation is almost the
same for all SMCMC algorithms because of the effectiveness of mHMC, the poor performance 
of SmHMC in estimating $\log Z_k$ shows that the independent MH kernel based on prior 
as proposal, which is employed in joint draw of SmHMC, is inefficient in this high-dimensional example. 
We also note that SmHMC (LEDH) and SmHMC (EDH) perform similarly to the PFPF (LEDH) and PFPF (EDH) ~\cite{li2017} for the same number of particles. However, estimation of $\log Z_k$ can be improved
significantly by increasing the number of particles in PFPF (EDH), with negligible computational overhead. The BPF suffers from severe 
weight degeneracy and shows poor estimation performance, even if a large number of particles is employed.

\begin{table}[htbp]
\small
\setlength{\tabcolsep}{1pt}
\centering
\caption{Average MSE, $\text{MSE}_{\log Z}^{(rel)}$, ESS(if applicable) and execution time per step in 
the linear Gaussian example.}
\label{tab:linear}
\begin{tabular}{|c|c|c|c|c|c|}
\hline
Algorithm      & \begin{tabular}[c]{@{}c@{}} No. of\\ particles\end{tabular}  & \begin{tabular}[c]{@{}c@{}}Avg.\\ MSE\end{tabular} & \begin{tabular}[c]{@{}c@{}}Avg.\\ $\text{MSE}_{\log Z}^{(rel)}$\end{tabular} & \begin{tabular}[c]{@{}c@{}}Avg.\\ ESS\end{tabular} & \begin{tabular}[c]{@{}c@{}}Exec.\\ time (s)\end{tabular} \\ \hline
KF             & N/A         & 0.07                                               & 0                                                                     & -                                                  & 0.002                                                    \\ \hline
SmHMC (LEDH) & 200         & 0.08                                               & 0.0006                                                                 & -                                                  & 2.5                                                     \\ \hline
SmHMC (EDH)  & 200         & 0.08                                               & 0.0006                                                                & -                                                  & 0.70                                                     \\ \hline
SmHMC          & 200         & 0.09                                               & 1.629                                                                 & -                                                  & 0.6                                                     \\ \hline
PFPF (LEDH)    & 200         & 0.09                                               & 0.0005                                                                & 25.1                                               & 1.90                                                     \\ \hline
PFPF (EDH)     & 200         & 0.09                                               & 0.0006                                                                & 21.7                                               & 0.015                                                    \\ \hline
PFPF (EDH)     & $10^4$      & 0.08                                               & 0.0001                                                                & 852                                                & 0.2                                                     \\ \hline
BPF            & 200         & 1.10                                               & 2.813                                                                 & 1.04                                               & 0.001                                                    \\ \hline
BPF            & $10^6$      & 0.20                                               & 0.0265                                                                & 1.62                                               & 2.5                                                     \\ \hline
\end{tabular}
\end{table}

\subsection{Nonlinear model 
with GMM process and measurement noises}
\label{sec:setup_multi_modal}
Here we consider a nonlinear dynamical model 
$g_k: \mathbb{R}^{d} \rightarrow
\mathbb{R}^{d}$ and measurement function $h_k: \mathbb{R}^{d} \rightarrow
\mathbb{R}^{d}$. The $c$-th element of the measurement vector is $h_k^c(x_k) = \frac{{(x_k^c)}^2}{20}, \text{ } \forall 1 \leq c \leq d$ .
$c$-th element of the state vector is defined as follows:
\begin{equation}
\begin{split}
  g_k^c(x_{k-1}) = 0.5x_{k-1}^c + 8\cos(1.2(k-1)) \\+
\begin{cases}
   2.5\frac{x_{k-1}^{c+1}}{1+{(x_{k-1}^c)}^2} & \text{, if $c=1$} \\ 
   2.5\frac{x_{k-1}^{c+1}}{1+{(x_{k-1}^{c-1})}^2} & \text{, if $1< c< d$} \\
   2.5\frac{x_{k-1}^c}{1+{(x_{k-1}^{c-1})}^2} & \text{, if $c=d$}
\end{cases}
\end{split} 
\end{equation} 
Process noise $v_k\sim\sum^3_{m=1}\tfrac{1}{3} \mathcal{N}(\mu_m\mathbf{1}_{d\times 1},\sigma_v^2I_{d\times d})$, 
with $\mu_1 =-1$, $\mu_2 =0$, $\mu_3 =1$ and $\sigma_v = 0.5$, 
and measurement noise $w_k \sim \sum^3_{n=1} \tfrac{1}{3}\mathcal{N}(\delta_n\mathbf{1}_{d\times 1},\sigma_w^2I_{d\times d})$, 
with $\delta_1 =-3$, $\delta_2 =0$, $\delta_3 =3$ and $\sigma_w = 0.1$.
The true state starts at
$x_0 = \mathbf{0}$. For all the filters, we use $p(x_0) = \mathcal{N}(\mathbf{0}_{d\times 1},I_{d\times d})$. 
The experiment is executed 100 times for 50 time steps. 
We perform two different experiments with $d = 144$ and $d = 400$.

\begin{table}[htbp]
\small
\setlength{\tabcolsep}{1pt}
\centering
\caption{Average MSE, acceptance rates(if applicable) and execution time per step in the nonlinear model 
with GMM process and measurement noises, based on 100 simulation trials.}
\label{tab:gmm}
{
\begin{tabular}{|c|c|c|c|c|c|c|c|}
\hline
\multirow{2}{*}{$d$} & \multirow{2}{*}{Algorithm} &  \multirow{2}{*}{\begin{tabular}[c]{@{}c@{}} No. of\\ particles\end{tabular}} &
\multirow{2}{*}{\begin{tabular}[c]{@{}c@{}}Avg.\\ MSE\end{tabular}} & \multicolumn{3}{|c|}{Acceptance rate} & \multirow{2}{*}{\begin{tabular}[c]{@{}c@{}}Exec.\\ time (s)\end{tabular}}\\\cline{5-7}
& & & & $\rho_1$ & $\rho_2$ & $\rho_3$ &\\\hline

\multirow{14}{*}{144} & \begin{tabular}[c]{@{}c@{}}SmHMC-GMM\\  (LEDH)\end{tabular} & 100 & 0.10 & 0.072 & 0.17 & 0.74 & 7.12\\\cline{2-8}
& SmHMC-GMM & 100 & 0.12 & 0.005 & 0.17 & 0.76 & 3.3\\\cline{2-8}
& PFPF-GMM & 200 & 0.10 & - & - & - & 7.4\\\cline{2-8}
& PF-GMM & \begin{tabular}[c]{@{}c@{}} 50 per\\ comp.\end{tabular} & 0.18 & - & - & - & 12.4\\\cline{2-8}
& GSPF & \begin{tabular}[c]{@{}c@{}} $10^4$ per\\ comp.\end{tabular} & 4.53 & - & - & - & 3.7\\\cline{2-8}
& EKF-GMM & N/A & 2.12 & - & - & - & 0.05\\\cline{2-8}
& UKF & N/A & 1.30 & - & - & - & 0.15\\\cline{2-8}
& LEDH & 500 & 9.05 & - & - & - & 13.9\\\cline{2-8}
& EDH & 500 & 11.54 & - & - & - & 0.04\\\cline{2-8}
& PFPF (LEDH) & 500 & 5.90 & - & - & - & 19\\\cline{2-8}
& PFPF (EDH) & $10^5$ & 3.01 & - & - & - & 4.60\\\cline{2-8}
& BPF & $10^6$ & 0.94 & - & - & - & 9\\\hline\hline
\multirow{14}{*}{400} & \begin{tabular}[c]{@{}c@{}}SmHMC-GMM\\  (LEDH)\end{tabular} & 100 & 0.09 & 0.018 & 0.096 & 0.60 & 59.5\\\cline{2-8}
& SmHMC-GMM & 100 & 0.11 & 0.005 & 0.085 & 0.64 & 13.3\\\cline{2-8}
& PFPF-GMM & 200 & 0.11 & - & - & - & 80.2\\\cline{2-8}
& PF-GMM & \begin{tabular}[c]{@{}c@{}} 50 per\\ comp.\end{tabular} & 0.11 & - & - & - & 142.7\\\cline{2-8}
& GSPF & \begin{tabular}[c]{@{}c@{}} $10^4$ per\\ comp.\end{tabular} & 5.17 & - & - & - & 9.4\\\cline{2-8}
& EKF-GMM & N/A & 1.61 & - & - & - & 0.3\\\cline{2-8}
& UKF & N/A & 5.18 & - & - & - & 1.37\\\cline{2-8}
& LEDH & 500 & 23.42 & - & - & - & 152.6\\\cline{2-8}
& EDH & 500 & 30.45 & - & - & - & 0.42\\\cline{2-8}
& PFPF (LEDH) & 500 & 16.38 & - & - & - & 194\\\cline{2-8}
& PFPF (EDH) & $10^5$ & 12.27 & - & - & - & 16.8\\\cline{2-8}
& BPF & $10^6$ & 1.31 & - & - & - & 26.2\\\hline
\end{tabular}}
\end{table}

Table~\ref{tab:gmm} shows that while
the proposed SmHMC-GMM (LEDH) achieves the same smallest average
MSE as the PFPF-GMM \cite{pal2018} algorithm among all evaluated methods
in the 144 dimensional scenario, the SmHMC-GMM (LEDH) leads to the smallest 
average MSE in the 400 dimensional scenario.
The SmHMC-GMM is an SmHMC variant such that after sampling of $(d_k,c_k)$ using 
the same distributions as in SmHMC-GMM  (LEDH) in the joint draw, $x_k$ is proposed using the particular component of the dynamic model specified by $d_k$. The comparison of 
acceptance rates in the joint draw and MSE for these
two algorithms shows that the use of particle flow in
the SmHMC-GMM (LEDH) method allows it to explore the state space more efficiently
than the SmHMC-GMM algorithm. 

The PF-GMM \cite{pal2017},
 which uses a separate LEDH filter to track each component
 of the posterior, performs reasonably well, whereas the particle 
 flow algorithms LEDH, EDH and the particle flow particle filters perform poorly as they are better suited for uni-modal posterior distributions. The Gaussian sum particle filter (GSPF) \cite{kotecha2003}, which approximates each component of the predictive and posterior densities by a Gaussian distribution by performing importance sampling exhibits poor representation capability in higher dimensions. The extended Kalman filter for GMM noises (EKF-GMM) and UKF lead to large estimation errors. The BPF also has high MSE even with $10^6$ particles, due to the weight degeneracy in the high-dimensional state space. 


\section{Conclusion}
\label{sec:conclusion}

In this paper, we proposed a series of composite MH kernels for SMCMC methods. These kernels are constructed based on invertible particle flow to achieve efficient exploration of high-dimensional state spaces. The EDH-based SMCMC method provides minimal computational overhead but a significant increase of the acceptance rate in the joint draw compared to the state-of-the-art SMCMC algorithm, the SmHMC method. For multi-modal distributions, a Gaussian mixture model-based particle flow is incorporated to migrate samples into high posterior density regions. Theoretical convergence results are also derived for the SMCMC methods.

We evaluated the proposed algorithms in three simulation examples. In the large spatial sensor network setup with high-dimensional
non-Gaussian distributions, the EDH and the LEDH methods provide the smallest MSEs, but there are no convergence results for these flow-based algorithms. The SmHMC methods based on the EDH or the LEDH flow have been shown in the paper to converge to the target distribution, and in the 400-dimensional filtering example, they provide the smallest MSEs among all particle filters and SmHMC algorithms for which convergence results have been established. In the linear Gaussian example, both the SmHMC (EDH) and the SmHMC (LEDH) methods provide smaller estimation errors of the normalisation constants compared to SmHMC. In the third example with high-dimensional nonlinear HMM models and GMM process and measurement noise, the proposed SmHMC-GMM (LEDH) algorithm provides the smallest estimation errors in both the 144-dimensional and 400-dimensional experiment settings.

An important  future  research  direction  is  a more extensive  experimental evaluation  of  flow-based  SMCMC  algorithms. These experiments should investigate the impact of the  initial  state  values,  the process  noise variance,  the measurement  noise  variance,  coupling in the  dynamic  models,  partial  observations,  the data  rate  and measurement uncertainty. Such experiments  can  shed  light  on  the robustness of  the  invertible  particle flow-based approach in practical applications and motivate the development of new algorithms that address any exposed deficiencies. One example of the real-world challenges in sensor networks  is  data  incest  due  to  the  inadvertent  re-use  of  the  same  measurements~\cite{fantacci2018}, which can be mitigated to some extent by a data incest management strategy that takes into account the network topology~\cite{brehard2007} or information fusion techniques with copula processes~\cite{liu2017}. 

Other problems common in practical sensor networks are missed detections, false alarms, and finite measurement resolution. Addressing such practical problems is important and can lead to interesting research directions, e.g., designing an appropriate combination of sequential MCMC, particle flow and random finite sets.

Beyond experimentation, there are important methodological and theoretical issues to explore in future work.  In terms of particle flow, the algorithms in this paper use deterministic flows in order to obtain an invertible mapping. It is more challenging to incorporate stochastic particle flow, but the stochastic flow algorithms have been demonstrated to achieve considerably better performance~\cite{daum2017}, so integration is desirable. The stiffness of the differential equations is an issue in both particle flow and Hamiltonian Monte Carlo and can lead to numerical instability. For this reason, and also to guarantee the invertible mapping property for the flows, we use a very small step size in the particle flow procedure. This leads to a greater computational overhead. Alternative strategies for addressing stiffness have been proposed in~\cite{daum2014b,daum2017}, and it would be interesting to explore their incorporation in the flow-based SMCMC framework. This  paper  derives  the  asymptotic  convergence  results  of  the  proposed  algorithms. Finite sample analysis of filter errors is an important direction to explore; \cite{finke2018} provides a valuable finite sample bound for SMCMC errors. In a similar manner to~\cite{cheng2017}, it may be possible to identify a relationship between the magnitude of the error and the stiffness of the flow.

\appendices
\section{Proof of Theorem~\ref{theorem:empirical_avg}}
\label{appendix:proof_theo1}
We start with several notations and propositions. For $k > 0$, 
the proposal distribution $q_k^{opt}$ that minimizes the variance of importance weights is the conditional density
of $x_k$ given $x_{0:k-1}$ under $\pi_k$ \cite{brockwell2010}.
\begin{align}
q_k^{opt}(x_{0:k-1},x_k) &= \overline{\pi}_k(x_{0:k-1},x_k):=\dfrac{\pi_k(x_{0:k})}{\pi_k(x_{0:k-1})}\,,\nonumber\\
&= p(x_k|x_{k-1},z_k)\,,\label{eq:pi_overline}
\end{align}
where $\pi_k(x_{0:k-1})=\int_{E}\pi_k(x_{0:k})dx_k$.

With this ``optimal'' density, the optimal importance weight $w_k^{opt}$ does not depend on $x_k$.
\begin{align}
w_k^{opt}(x_{0:k}) &\propto \pi_{k/k-1}(x_{0:k-1})\,,\nonumber\\&:=\dfrac{\pi_k(x_{0:k-1})}{\pi_{k-1}(x_{0:k-1})}= p(z_k|x_{k-1})\,.\label{eq:pi_k_slash_k_1}
\end{align}


In the SMCMC algorithm, at iteration $i$ of any given time step $k$, 
there is a joint draw of $x_{0:k}^i$ which is then followed by 
individual refinements of $x_{0:k}^{(i)}$ using the 
\emph{Metropolis within Gibbs} technique, if $k > 0$.
There is no refinement step at time $k=0$. Using Equations~\eqref{eq:recursive_joint_two},
~\eqref{eq:approx_posterior}, ~\eqref{eq:w_0_def} and ~\eqref{eq:w_k_def},
we calculate the acceptance probability of the joint draw as follows:
\begin{align}
\alpha_0(x_{0}^{i-1},x_{0}^{*(i)}) =& \min\left(1,\dfrac{\pi_0(x_{0}^{*(i)})q_0(x_{0}^{i-1})}
{q_0(x_{0}^{*(i)}) \pi_0(x_{0}^{i-1})}\right)\,,\nonumber\\
=& \min\left(1,\dfrac{w_0(x_{0}^{*(i)})}{w_0(x_{0}^{i-1})}\right)\,,
\label{eq:alpha_0}
\end{align}
and for $k > 0$,
\begin{align}
&\alpha_k(x_{0:k}^{i-1},x_{0:k}^{*(i)}) \nonumber\\
&= \min\left(1,\dfrac{\breve \pi_k(x_{0:k}^{*(i)})\widehat{\pi}_{k-1}^{(N_p)}(x_{0:k-1}^{i-1})q_k(x_{0:k-1}^{i-1},x_{k}^{i-1})}
{\widehat{\pi}_{k-1}^{(N_p)}(x_{0:k-1}^{*(i)})q_k(x_{0:k-1}^{*(i)},x_{k}^{*(i)})\breve \pi_k(x_{0:k}^{i-1})}\right)\,,\nonumber\\
&= \min\left(1,\dfrac{\pi_k(x_{0:k}^{*(i)})\pi_{k-1}(x_{0:k-1}^{i-1})q_k(x_{0:k-1}^{i-1},x_{k}^{i-1})}
{\pi_{k-1}(x_{0:k-1}^{*(i)})q_k(x_{0:k-1}^{*(i)},x_{k}^{*(i)}) \pi_k(x_{0:k}^{i-1})}\right)\,,\nonumber\\
&= \min\left(1,\dfrac{w_k(x_{0:k}^{*(i)})}{w_k(x_{0:k}^{i-1})}\right)\,.
\label{eq:alpha_k}
\end{align}

We define the independent MH kernel to initialize the algorithm,
$K_0^{draw}: E_0 \times \mathcal{F}_0 \rightarrow [0, 1]$:
\begin{align}
K_0^{draw}(x_0,dx_0^{'}) =& \alpha_0(x_{0},
x_{0}^{'})q_0(dx_{0}^{'})\nonumber\\
& + \left(1-\int_{E_0} \alpha_0(x_{0},y_{0})q_0(dy_0)\right)\delta_{x_{0}}(d x_{0}^{'})\,.
\end{align}
For $k > 0$, we use $\widehat{\pi}_{k-1}^{(N_p)} \in \mathcal{P}_{k-1}(E_{k-1})$ to 
construct the joint proposal $(\widehat{\pi}_{k-1}^{(N_p)} \times q_k)$ for
the joint draw, which is associated with the Markov kernel
$K_k^{draw}: E_k \times \mathcal{F}_k \rightarrow [0, 1]$, defined by
\begin{align}
&K_k^{draw}(x_{0:k},dx_{0:k}^{'}) = \alpha_k(x_{0:k},
x_{0:k}^{'})(\widehat{\pi}_{k-1}^{(N_p)} \times q_k)(dx_{0:k}^{'})\nonumber\\
& + \left(1-\int_{E_k} \alpha_k(x_{0:k},y_{0:k})(\widehat{\pi}_{k-1}^{(N_p)}\times q_k)(dy_{0:k})\right)\delta_{x_{0:k}}(d x_{0:k}^{'})\,.
\end{align}

\begin{lemma}
Given Assumption~\ref{assump:bound}, $K_0^{draw}(x_0,dx_0^{'})$
is an independent Metropolis kernel, 
uniformly ergodic of invariant distribution $\pi_0(dx_1)$.
\label{lemma:K_0}
\end{lemma}

\begin{proof}
From Equation \eqref{eq:alpha_0}, we see that $K_0^{draw}$ is an 
independent Metropolis kernel with target distribution $\pi_0$ and
proposal $q_0$.  If Assumption~\ref{assump:bound} is satisfied, uniform 
ergodicity follows from Corollary 4 in \cite{tierney1994}. 
\end{proof}

\begin{proposition}
Given Assumption~\ref{assump:bound}, for any $k > 0$, 
$K_k^{draw}(x_{0:k},d x_{0:k}^{'})$ is uniformly  
ergodic of invariant distribution
\begin{align}
\breve {\pi}_{k}^{(N_p)}(d x_{0:k})
= \dfrac{\pi_{k/k-1}(x_{0:k-1})\cdot(\widehat{\pi}_{k-1}^{(N_p)}\times\overline{\pi}_{k})(d x_{0:k})}{\widehat{\pi}_{k-1}^{(N_p)}(\pi_{k/k-1})}\,,
\label{eq:breve_pi_def}
\end{align}
where $\overline{\pi}_k(x_{0:k-1},dx_k)$ and $\pi_{k/k-1}(x_{0:k-1})$
are defined by Equations~\eqref{eq:pi_overline} and~\eqref{eq:pi_k_slash_k_1}, respectively.
\label{prop:K_k_draw_ergodic}
\end{proposition}
\begin{proof}
From Equation~\eqref{eq:alpha_k}, we see that $K_k^{draw}$ is an independent 
Metropolis kernel with target distribution $\breve{\pi}_k^{(N_p)}$ and proposal $\widehat{\pi}_{k-1}^{(N_p)} \times q_k$.
From Equations~\eqref{eq:recursive_joint_two}, ~\eqref{eq:approx_posterior}, ~\eqref{eq:pi_overline} and ~\eqref{eq:pi_k_slash_k_1} we have
\begin{align}
&\breve {\pi}_{k}^{(N_p)}(x_{0:k}) \nonumber\\
=& \dfrac{\pi_{k/k-1}(x_{0:k-1})\overline{\pi}_{k}(x_{0:k-1},x_k)\widehat{\pi}_{k-1}^{(N_p)}(x_{0:k-1})}
{\int_{E_k}\pi_{k/k-1}(x_{0:k-1})\overline{\pi}_{k}(x_{0:k-1},x_k)\widehat{\pi}_{k-1}^{(N_p)}(x_{0:k-1})d x_{0:k}}\,,\nonumber\\
=& \dfrac{\pi_{k/k-1}(x_{0:k-1})\cdot(\widehat{\pi}_{k-1}^{(N_p)} \times \overline{\pi}_{k})(x_{0:k})}
{\int_{E_{k-1}}\pi_{k/k-1}(x_{0:k-1})\widehat{\pi}_{k-1}^{(N_p)}(x_{0:k-1})d x_{0:k-1}}\,,\nonumber\\
=& \dfrac{\pi_{k/k-1}(x_{0:k-1})\cdot(\widehat{\pi}_{k-1}^{(N_p)} \times \overline{\pi}_{k})(x_{0:k})}
{\widehat{\pi}_{k-1}^{(N_p)}(\pi_{k/k-1})}\,.
\end{align}
\end{proof}
We denote the cumulative MCMC kernel of all the refinement steps 
by $K_{k}^{refine}: E_k \times \mathcal{F}_k \rightarrow [0, 1]$ 
for $k > 0$. We note that $K_{k}^{refine}$ also has the same invariant 
distribution $\breve{\pi}_{k}^{(N_p)}$, like $K_{k}^{draw}$. We define 
the overall MCMC kernel for SMCMC, $K_{k}: E_k \times \mathcal{F}_k \rightarrow [0, 1]$ for $k \ge 0$ as follows,
\begin{align}
K_0(x_0,dx_0^{'}) := K_0^{draw}(x_0,dx_0^{'})\,,
\label{eq:K_0_def}
\end{align} 
and for $k > 0$,
\begin{align}
&K_k(x_{0:k},dx_{0:k}^{'}) := K_k^{refine}K_k^{draw}(x_{0:k},dx_{0:k}^{'})\,,\nonumber\\
&= \int_{E_k} K_k^{refine}(y_{0:k},dx_{0:k}^{'})K_k^{draw}(x_{0:k},dy_{0:k})\,.
\label{eq:K_k_def}
\end{align}

\begin{proposition}
For any $k \ge 0$, $K_k(x_{0:k},dx_{0:k}^{'})$ is uniformly 
ergodic of invariant distribution $\breve{\pi}_{k}^{(N_p)}(d x_{0:k})$.
\label{prop:K_k_ergodic}
\end{proposition}
\begin{proof}
For $k=0$, the result is trivially true from the definition of $K_0$ in 
Equation \eqref{eq:K_0_def} and Lemma \ref{lemma:K_0}. For $k > 0$, 
the assertion follows from application of Corollary 4 and 
Proposition 4 in \cite{tierney1994} to the definition of $K_k$ in Equation \eqref{eq:K_k_def}.
\end{proof}

\begin{proposition}
Suppose $\pi \in \mathcal{P}(E)$ and $K: E \times \mathcal{F} \rightarrow [0, 1]$ is an 
ergodic MCMC kernel of invariant distribution $\pi(dx)$. For any $f:E \rightarrow \mathbb{R}$, 
if $f \in \mathcal{L}^1(E,\mathcal{F},\pi)$, $\widehat{\pi}^{(N_p)}(f)$ 
converges to $\pi(f)$ almost surely, irrespective of the starting point of the Markov chain $x^{(0)}$.
\label{prop:mcmc_slln}
\end{proposition}
\begin{proof}
See Theorem 3 in \cite{tierney1994}.
\end{proof}

\subsection{Proof of Theorem~ \ref{theorem:empirical_avg}}
\begin{proof}
We prove the theorem using induction over $k$. For $k=0$, the theorem is 
trivially true which can be seen by applying Proposition \ref{prop:mcmc_slln} 
with Lemma \ref{lemma:K_0}. Let us assume that the theorem is true for $k-1$. We consider the following decomposition and examine each term individually.
\begin{align}
\widehat{\pi}_{k}^{(N_p)}(f_k) - \pi_k(f_k) =\,\,&[\widehat{\pi}_{k}^{(N_p)}(f_k) -\breve {\pi}_{k}^{(N_p)}(f_k)]\nonumber\\
&+ [\breve {\pi}_{k}^{(N_p)}(f_k) - \pi_k(f_k)]\,.
\label{eq:widehat_pi_diff_decomposition}
\end{align}
From Equation \eqref{eq:breve_pi_def}, we have
\begin{align}
\lim_{N_p \rightarrow \infty}\breve {\pi}_{k}^{(N_p)}(f_k) = \dfrac{\displaystyle\lim_{N_p \rightarrow \infty}\widehat{\pi}_{k-1}^{(N_p)}(\pi_{k/k-1}\bar{f}_k)}{\displaystyle\lim_{N_p \rightarrow \infty}\widehat{\pi}_{k-1}^{(N_p)}(\pi_{k/k-1})}\,,
\label{eq:omwga_integral}
\end{align}
where,
\begin{align}
\bar{f}_k(x_{0:k-1})= \int_{E} f_k(x_{0:k-1},x_k) \overline{\pi}_k(x_{0:k-1},dx_k)\,.
\label{eq:f_bar_defn}
\end{align}
We note that, from Equation \eqref{eq:pi_k_slash_k_1},
\begin{align}
\pi_{k-1}(\pi_{k/k-1}) &= \int_{E_{k-1}} \pi_k(dx_{0:k-1}) \,,\nonumber\\
&=\int_{E_k} \pi_k(dx_{0:k}) = 1\,,
\label{eq:pi_k_1_den}
\end{align}
and from Equation \eqref{eq:f_bar_defn},
\begin{align}
\pi_{k-1}(\pi_{k/k-1}\bar{f}_k) &= \int_{E_{k-1}} \bar{f}_k(x_{0:k-1}) \pi_k(dx_{0:k-1}) \,,\nonumber\\
&= \int_{E_{k}} f_k(x_{0:k}) \pi_k(dx_{0:k}) \,,\nonumber\\
&= \pi_k(f_k) \leq  \pi_k(|f_k|) < \infty \,,
\label{eq:pi_k_1_num}
\end{align}
because $f_k \in \mathcal{L}^1(E_k, \mathcal{F}_k,\pi_k)$.  As 
Theorem \ref{theorem:empirical_avg} holds for $k-1$, we have 
from Equations~\eqref {eq:pi_k_1_den} and~\eqref {eq:pi_k_1_num}
\begin{align}
\widehat{\pi}_{k-1}^{(N_p)}(\pi_{k/k-1}) \longrightarrow 1\,,
\end{align}
 and 
\begin{align}
\widehat{\pi}_{k-1}^{(N_p)}(\pi_{k/k-1}\bar{f}_k) \longrightarrow \pi_k(f_k)\,,
\end{align}
almost surely, as $N_p \rightarrow \infty$ for any $x_{0:l}^{(0)} \in \mathcal{S}_l$ 
for $0 \leq l \leq k-1$. Applying Equation~\eqref{eq:omwga_integral},
this implies that
\begin{align}
\lim_{N_p \rightarrow \infty}\breve {\pi}_{k}^{(N_p)}(f_k) - \pi_k(f_k) = 0 \,,
\label{eq:convergence_part2}
\end{align}
almost surely.

In the SMCMC algorithm, at time step $k$, the MCMC kernel used 
to sample $\{x_{0:k}^{i}\}_{i=1}^{N_p}$ is $K_k$, which, from Proposition~\ref{prop:K_k_ergodic}, is 
uniformly ergodic of invariant distribution $\breve
{\pi}_{k}^{(N_p)}$. The applicability of Proposition~\ref{prop:mcmc_slln} 
depends on the integrability of $f_k$ w.r.t.\ $\displaystyle\lim_{N_p \rightarrow \infty}\breve {\pi}_{k}^{(N_p)}$. 
From~\eqref{eq:convergence_part2}, if $f_k \in \mathcal{L}^1(E_k, \mathcal{F}_k,\pi_k)$, 
we also have $f_k \in \mathcal{L}^1(E_k, \mathcal{F}_k,\displaystyle\lim_{N_p \rightarrow \infty}\breve {\pi}_{k}^{(N_p)})$ 
with probability 1, which allows us to use Proposition \ref{prop:mcmc_slln} to obtain
\begin{align}
\widehat{\pi}_{k}^{(N_p)}(f_k) - \displaystyle\lim_{N_p \rightarrow \infty}\breve {\pi}_{k}^{(N_p)}(f_k) \longrightarrow 0\,,
\label{eq:convergence_part1}
\end{align}
almost surely, as $N_p \rightarrow \infty$, for any $x_{0:k}^{(0)} \in \mathcal{S}_k$. 
We use the equivalence in \eqref{eq:convergence_part2} and the almost sure convergence in
\eqref{eq:convergence_part1} in
Equation~\eqref{eq:widehat_pi_diff_decomposition} to complete the
proof of the theorem.
\end{proof}

\section {Proof of Corollary~ \ref{cor:Z_k_ratio}}
\label{appendix:proof_corr}
\begin{proof}
From Equation \eqref{def:Z_0}, we have $Z_0 = q_0(w_0)$,
where $w_0 \in \mathcal{L}^1(E_0, \mathcal{F}_0,q_0)$, 
because of Assumption \ref{assump:bound}. A straightforward application 
of \emph{Kolmogorov's Strong Law of Large Numbers} (SLLN) proves the Corollary for $k=0$. \\
For $k > 0$, from Equation \eqref{def:Z_k_ratio}, we have
\begin{align}
\dfrac{Z_k}{Z_{k-1}} &= (\pi_{k-1}\times q_k)(w_k) \leq B_k \nonumber\,,
\end{align}
because of Assumption \ref{assump:bound}. As $\{x_{0:k}^{*(i)}\}_{i=1}^{N_p}$ are 
i.i.d samples from $\widehat{\pi}_{k-1}^{(N_p)} \times q_k$, we have, from \eqref{eq:Z_k_ratio}
\begin {align}
\reallywidehat{\Bigg(\dfrac{Z_k}{Z_{k-1}}\Bigg)}^{(N_p)}
&\longrightarrow \left(\displaystyle\lim_{N_p \rightarrow \infty}
\widehat{\pi}_{k-1}^{(N_p)} \times q_k\right)(w_k) \,\nonumber\\
&= (\pi_{k-1} \times q_k)(w_k) =\dfrac{Z_k}{Z_{k-1}} \,\nonumber
\end{align}
almost surely, as $N_p \rightarrow \infty$ for any $x_{0:l}^{(0)} \in \mathcal{S}_l$ for $0 \leq l \leq k$. The 
first step follows from Kolmogorov's SLLN and the second step follows from Theorem \ref{theorem:empirical_avg}.
\end{proof}

\section*{Acknowledgment}

We acknowledge the support of the Natural Sciences and Engineering Research Council of Canada (NSERC) [2017-260250].

\bibliographystyle{IEEEtran}
\bibliography{TSP_SMCMC}

%
%
%
%
%




\begin{IEEEbiography}
[{\includegraphics[width=1in,height=1.25in,clip,keepaspectratio]{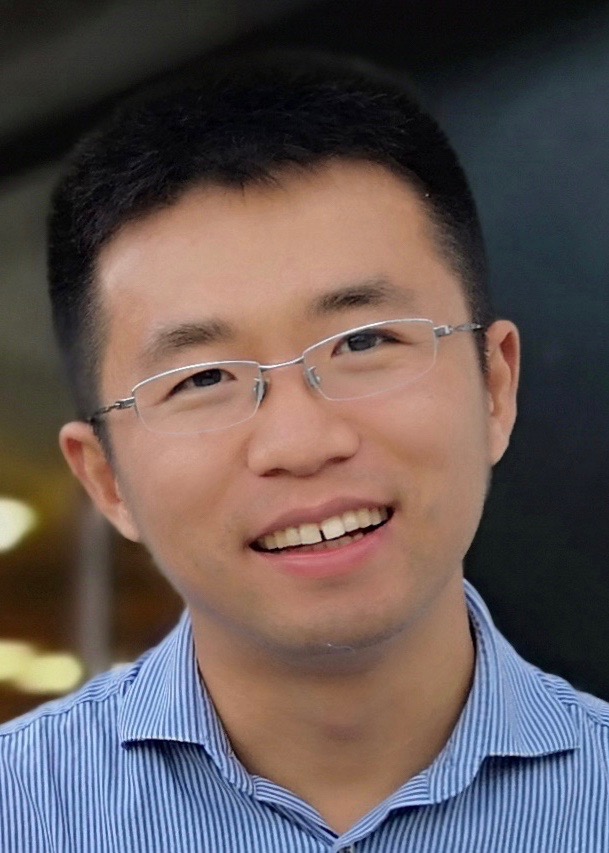}}]{Yunpeng Li}(S’16) received the B.A. and M.S. Eng. degrees from the Beijing University of Posts and Telecommunications, Beijing, China, in 2009 and 2012, respectively, and the Ph.D. degree in the Department of Electrical and Computer Engineering at McGill University in Montreal, Quebec, Canada in 2017. From 2017 to 2018, he was Postdoctoral Research Assistant in Machine Learning at the Machine Learning Research Group, University of Oxford, Oxford, U.K. He was Junior Research Fellow at the Wolfson College, University of Oxford in 2018. Since August 2018, he has been Lecturer in Artificial Intelligence in the Department of Computer Science at the University of Surrey, Guildford, U.K. His research interests include Bayesian inference, Monte Carlo methods, and statistical machine learning.
\end{IEEEbiography}

\begin{IEEEbiography}
[{\includegraphics[width=1in,height=1.25in,clip,keepaspectratio]{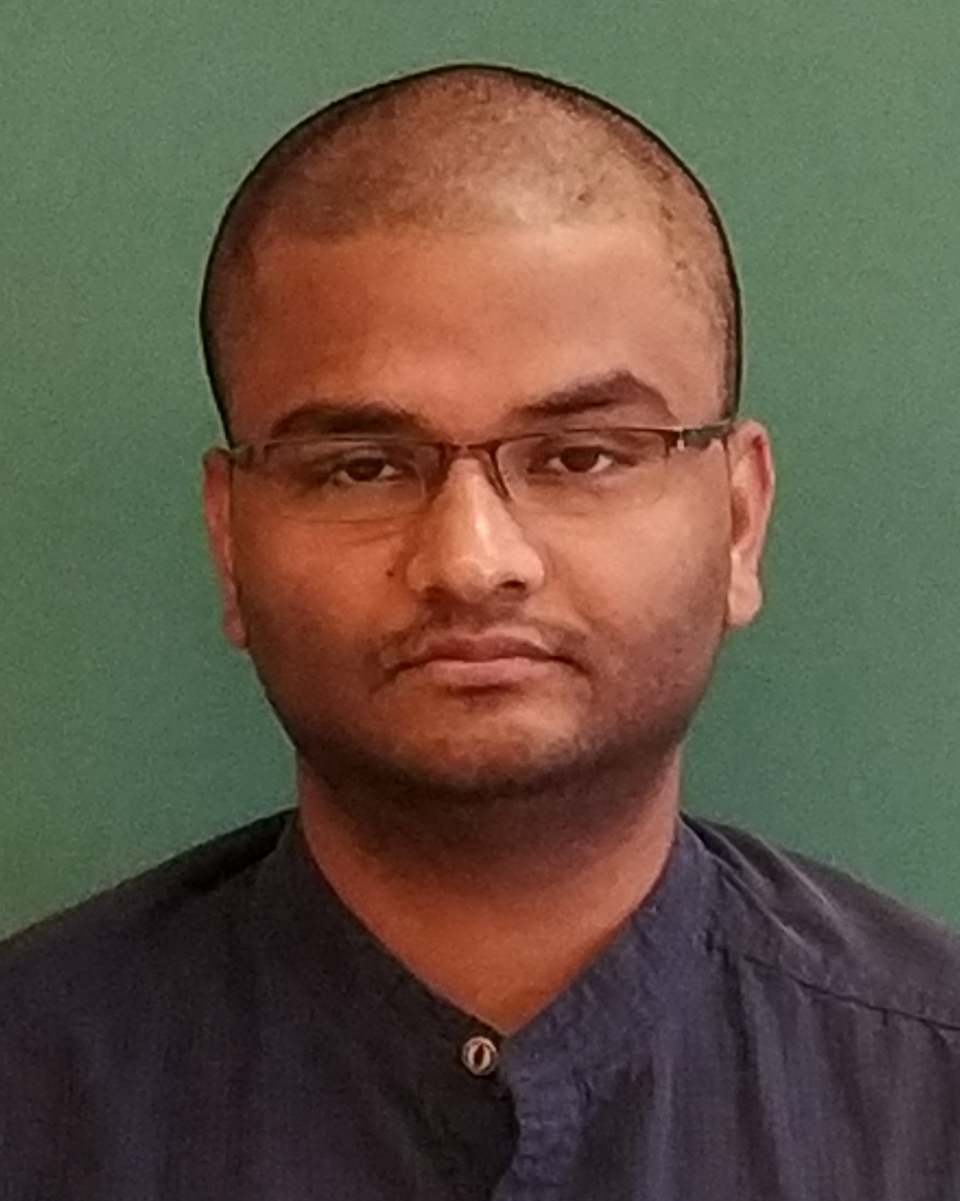}}]{Soumyasundar Pal} received the B.E. degree from the Jadavpur University, Kolkata, India in 2012 and the M.E. degree from the Indian Institute of Science, Bangalore, India, in 2014. Since 2016, he has been a Ph.D. student in the Department of Electrical and Computer Engineering at McGill University in Montreal, Quebec, Canada. His research interests include Bayesian inference, Monte Carlo methods, and machine learning on graphs.
\end{IEEEbiography}

\begin{IEEEbiography}
[{\includegraphics[width=1in,height=1.25in,clip,keepaspectratio]{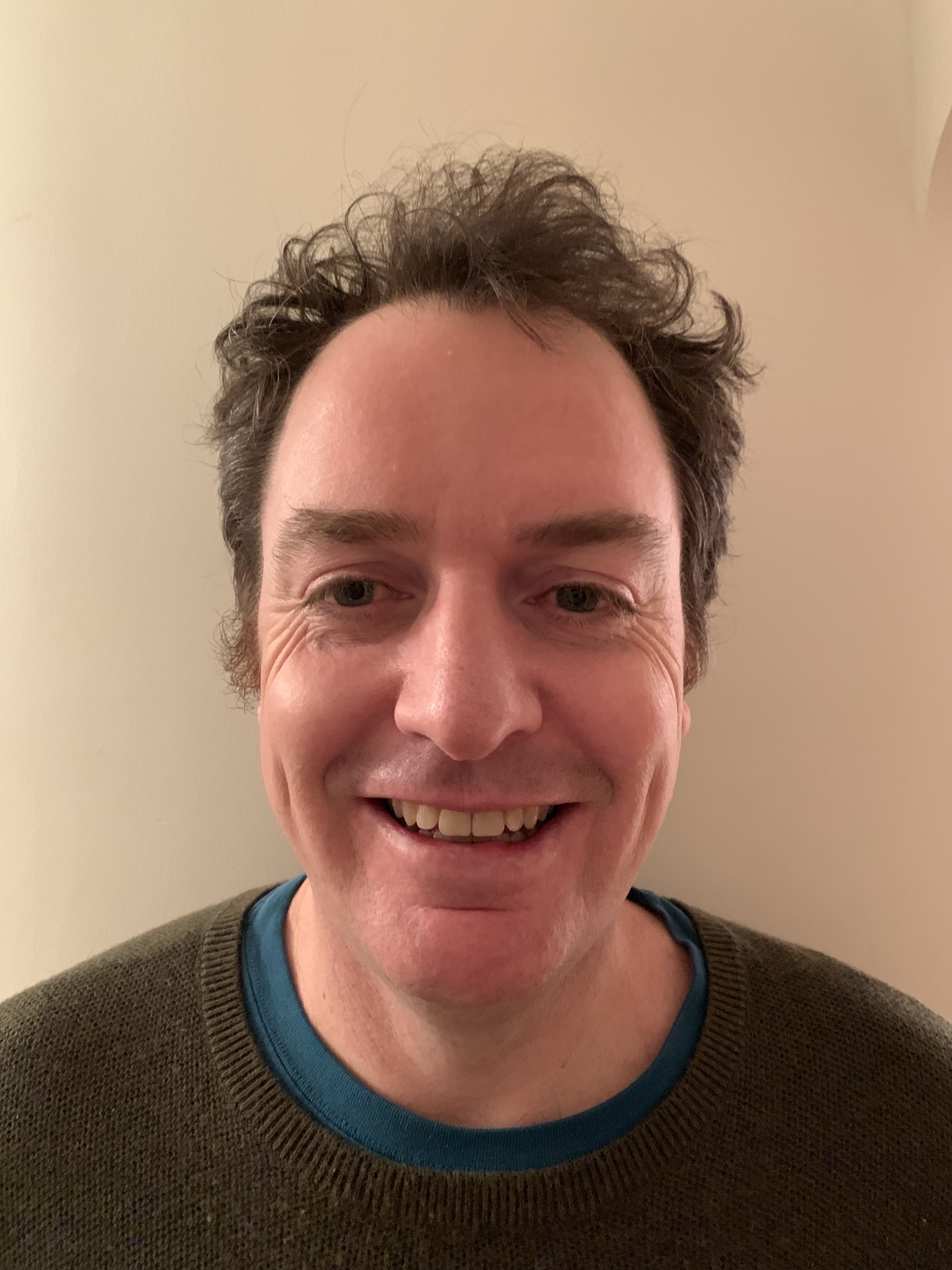}}]{Mark~J.~Coates}(M'00-SM’04) received the B.E.~degree in computer systems engineering from the University of Adelaide, Adelaide, Australia, in 1995, and the Ph.D. degree from the University of Cambridge, Cambridge, U.K., in 1999. He joined McGill University, Montreal, Canada, in 2002, where he is currently a Professor in the Department of Electrical and Computer Engineering. He was a research associate and a Lecturer at Rice University, TX, USA, from 1999 to 2001. From 2012 to 2013, he worked as a Senior Scientist at Winton Capital Management, Oxford, U.K. He currently serves as Associate Editor for \textsc{IEEE Transactions on Signal and Information Processing over Networks}. His research interests include statistical signal processing, machine learning, communication and sensor networks, and Bayesian and Monte Carlo inference. 
\end{IEEEbiography}

 \end{document}